\documentclass[11pt]{article}
\usepackage{graphicx}
\usepackage{color}
\usepackage{natbib}

{\makeatletter
 \gdef\xxxmark{%
   \expandafter\ifx\csname @mpargs\endcsname\relax % in minipage?
     \expandafter\ifx\csname @captype\endcsname\relax % in figure/caption?
       \marginpar{xxx}% not in a caption or minipage, can use marginpar
     \else
       xxx % notice trailing space
     \fi
   \else
     xxx % notice trailing space
   \fi}
 \gdef\xxx{\@ifnextchar[\xxx@lab\xxx@nolab}
 \long\gdef\xxx@lab[#1]#2{\textbf{[\xxxmark #2 ---{\sc #1}]}}
 \long\gdef\xxx@nolab#1{\textbf{[\xxxmark #1]}}
 \long\gdef\xxx@lab[#1]#2{}\long\gdef\xxx@nolab#1{}%
}

\usepackage{amsmath, amsthm, amssymb}

\usepackage[noend]{algorithmic}
\usepackage{algorithm}

 \newtheorem{theorem}{Theorem}[section]
 \newtheorem{lemma}[theorem]{Lemma}
 
 \newtheorem{corollary}[theorem]{Corollary}

 \newtheorem{definition}[theorem]{Definition}

\newenvironment{proofof}[1]{\begin{proof}[Proof of {#1}]}{\end{proof}}
\newcommand{\E}{\ensuremath{\mathbf{E}}}
\newcommand{\opt}{\textsf{Opt}\xspace}
\newcommand{\Pure}{\textsf{Pure}\xspace}
\newcommand{\Gold}{\textsf{Gold}\xspace}
\renewcommand{\Pr}{\ensuremath{\mathbf{Pr}}}

\newcommand{\NPclass}{\textsf{NP}\xspace}

\usepackage{fullpage}
\usepackage{xspace}

\newcommand{\eps}{\ensuremath{\varepsilon}}

\newcommand{\Alg}{\ensuremath{\mathsf{Alg}}\xspace}
\newcommand{\greedy}{\ensuremath{\mathsf{Greedy}}\xspace}
\newcommand{\elements}{\ensuremath{\mathcal E}\xspace}
\newcommand{\sets}{\ensuremath{\mathcal S}\xspace}
\newcommand{\cov}{\ensuremath{\mathcal C}\xspace}

{\addtolength{\itemsep}{-\baselineskip}}
{\addtolength{\abovecaptionskip}{-2\baselineskip}}
{\addtolength{\belowcaptionskip}{-\baselineskip}}

\title{Almost Optimal Streaming Algorithms for Coverage Problems}

\author{
  MohammadHossein Bateni\\Google Research
  \and Hossein Esfandiari\\University of Maryland
  \and Vahab Mirrokni\\Google Research
}

\date{}

\begin{document}
\sloppy
\maketitle

\begin{abstract}
Maximum coverage and minimum set cover problems---here collectively
called coverage problems---have been studied extensively in streaming models. 
However, previous research not only achieve suboptimal approximation factors and space complexities, but also
study a restricted set-arrival model which makes an explicit or implicit assumption on oracle
access to the sets, ignoring the complexity of reading and storing the whole set at once. 
In this paper, we address the above shortcomings, and present algorithms with improved approximation factor and improved space complexity, and prove that our results are almost tight. Moreover, unlike most of previous work, our results hold in a more general edge-arrival model. 

More specifically, consider an instance with $n$ sets, together covering $m$ elements.
Information arrives in the form of ``edges'' from sets to elements (denoting membership) in arbitrary order.
\begin{enumerate}
	\item We present (almost) optimal approximation algorithms for maximum coverage and minimum set cover problems in the streaming model with an (almost) optimal space complexity of $\tilde{O}(n)$; i.e., the space is {\em independent of the size of the sets or the size of the ground set of elements}. These results not only improve the best known algorithms for the set-arrival model, but also are the first such algorithms for the more powerful {\em edge-arrival} model.
	\item In order to achieve the above results, we introduce a new general sketching technique for coverage functions: One can apply this sketching scheme to convert an $\alpha$-approximation
	algorithm for a coverage problem to a
	$(1-\eps)\alpha$-approximation algorithm for the same problem in 
	streaming model.
	\item We show the significance of our sketching technique by ruling out the possibility of solving coverage problems via accessing (as a black box) a $(1 \pm \eps)$-approximate oracle (e.g., a sketch function) that estimates the coverage function 
	on any subfamily of the sets. Finally, we show that our streaming algorithms achieve an almost optimal space complexity.
\end{enumerate}

\end{abstract}

\thispagestyle{empty}
\newpage
\setcounter{page}{1}

\section{Introduction}
Maximum coverage and minimum set cover problems---here collectively
called {\em coverage problems}---are among the most fundamental
problems in optimization and computer science.  Coverage problems have
a variety of machine-learning and data-mining applications (for
examples in data summarization and web mining,
see~\cite{CKT10,AMT,pods12,nips13,karbasiKDD2014}).  Solving such
problems has become increasingly important for various real-world
large-scale data-mining applications where due to the sheer amount of
data, either the computation has to be done in a distributed
manner~\cite{CKT10,CKW10,BST12,KMVV13,nips13,IMMM14,MZ15}, or the data
is presented and needs to be analyzed in a
stream~\cite{karbasiKDD2014,assadi2016tight,saha2009maximum,nisan2002communication,emek2014semi,chakrabarti2015incidence}.

These problems have been explored extensively in the literature, but
despite development of several scalable algorithms, the
existing approaches still suffer from a few shortcomings.
First of all, most previously studied models make an explicit or
implicit assumption on having oracle access to each set in its entirety. This
assumption, in particular, ignores the computational complexity of
reading the whole set, or computing the marginal impact of adding a
subset to the solution (i.e., computing union and intersection of
family of subsets). For instance, in the streaming setting, this
assumption is implied in the extensively studied set-arrival
model~\cite{saha2009maximum,emek2014semi,demaine2014streaming,chakrabarti2015incidence}.
%In distributed models, this assumption is implied in usage of a value
%oracle access to the submodular
%function~\cite{BST12,KMVV13,nips13,karbasiKDD2014,MZ15,BENW15}.
Such models are less realistic since all the information of each set need to be gathered together. 
The set-arrival setting directly translates to the vertex-arrival setting in graph streaming%
\footnote{Modeled as a bipartite graph where vertices on one side corresponds to
the sets and vertices on the other side corresponds to elements. 
See Preliminaries for a formal definition.}, 
which is less interesting than the popular edge-arrival setting~\cite{ahn2013spectral,andoni2014towards,assadi2015tight,chitnis2015kernelization,esfandiari2015streaming,kapralov2014approximating,kapralov2015streaming}.
Secondly, current streaming algorithms often achieve
suboptimal approximation guarantees compared to the offline optimum or  do not have the best space
complexities in terms of the number of sets in the input.%
\footnote{We focus on the regime where the number of the element (i.e., the size of the ground set)
is significantly larger than the number of sets, hence the importance of having bounds in terms of
the number of sets rather than elements.\label{foot:nm-regime}}

In this paper, we aim to address the above issues. We develop
streaming algorithms that achieve optimal approximation guarantees as
well as optimal space complexities for coverage problems without any
oracle-access assumptions.  Moreover, our algorithm works in the
(more general) edge-arrival streaming model.
%Interestingly, we can
%achieve all these goals for a majority of coverage problems.
At the core of our analysis lies a simple, yet subtle sketching technique.
%To do so,
%we develop a sketching technique that provides us with the necessary
%tools for obtaining the results.
In order to demonstrate the power of this technique, we show why
natural sketching approaches do not work well. We also demonstrate
that oracle access to a noisy estimator for the coverage function is
not sufficient. We first present more formal definitions before elaborating 
on these results.

\subsection{Preliminaries}
\paragraph{Coverage Problems}
We study three related {\em coverage} problems. The setting includes
a ground set \elements of $m$ elements, and a family
$\sets\subseteq 2^\elements$ of $n$ subsets of the elements (i.e.,
$n=|\sets|$ and $m=|\elements|$).\footnote{There are two separate series of
  work in this area.  We use the convension of the submodular/welfare
  maximization formulation~\cite{badanidiyuru2012sketching}, whereas the hypergraph-based formulation~\cite{saha2009maximum} typically
  uses $n, m$ in the opposite way.} 
{\em The coverage function} $\cov$ is defined as
$\cov(S) = |\cup_{U\in S} U|$ for any subfamily $S\subseteq \sets$ of
subsets.
In the {\em $k$-cover} problem, given a parameter $k$, the goal is to
find $k$ sets in $\sets$ with the largest union size.  We sometimes
use $\opt_k$ to denote the size of the union for the optimum solution.
In the {\em set cover} problem, the goal is to pick the minimum number of
sets from $\sets$ such that all elements in $\elements$ are covered.
We also study a third problem: In the {\em set cover with $\lambda$
  outliers} problem\footnote{This is sometimes called the $(1-\lambda)$-%
  partial cover problem in the literature.},
the goal is to find the minimum number of sets covering at 
least a $1-\lambda$ fraction of the elements in \elements.

Coverage problems may be modeled as a bipartite graph $G$, 
where \sets corresponds to one part of the vertices, and
\elements corresponds to the other part. 
A vertex representing the set $S \in \sets$
has $\vert S\vert$ edges in $G$, one
to each element $i \in S$.  For simplicity, we assume that there
is no isolated vertex in \elements. 
For a subset $S$ of vertices in a graph $G$, let $\Gamma(G, S)$ denote the set of neighbors
of $S$. When $G$ is the graph corresponding to the original
coverage instance, we have $\cov(S) = |\Gamma(G, S)|$ if $S$ is a
subfamily of the sets $\sets$.

In the offline setting, a simple greedy algorithm achieves
$1-\frac 1 e$ approximation for $k$-cover and $\log m$ approximation
algorithm for the set cover problem.\footnote{Unless otherwise specified, we use the wide-spread
convension for approximation ratios: factors larger than one for minimization problems 
and factors smaller than one for maximization problems.}
Moreover, improving these
approximation factors
%for these problems is 
are impossible unless $\NPclass$ has  slightly superpolynomial time algorithm~\cite{feige1998threshold}.

\paragraph{Streaming models} 
%Our access to the input---i.e., sets and elements--- is not unrestricted. 
In the {\em streaming} model, we
focus on the so-called {\em edge-arrival} %(aka {\em element arrival})
model as opposed to the more studied {\em set-arrival} (aka {\em
  vertex-arrival}) model. In the former, edges arrive one by one, so
we get to know about the set-element membership relations one at a time,
whereas in the latter, sets arrive and bring with them a list of their
elements. The number of passes allowed for processing the data is
crucial and may change the nature of the problem.

\paragraph{The $(1 \pm \eps)$-approximate oracle.} We say $\cov_\eps$ is a
{\em $(1 \pm \eps)$-approximate oracle} to coverage function $\cov$ if, given a subfamily
of sets, it gives us an estimate of their union size within $1\pm\eps$
precision. In other words, $\cov_\eps$ estimates the coverage
function $\cov$ on any subfamily of the sets as a black box; i.e.,
for any subset $S\subseteq \sets$, we have
\begin{align*}
  (1-\epsilon)\cov_{\epsilon}(S) \leq \cov(S) \leq
  (1+\epsilon)\cov_{\epsilon}(S). 
\end{align*}

%Finally, there is a recent paper in the streaming model~\cite{karbasiKDD2014} in which the authors present a streaming $1/2$-approximation algorithm with one pass and linear memory. 
%$k$-cover problem is fairly well studied in the classical setting. 

\subsection{Related work}
%{\bf \noindent Streaming models for coverage problems.}
Coverage problems have been studied extensively in the context of
set-arrival
models~\cite{assadi2016tight,saha2009maximum,nisan2002communication,emek2014semi,chakrabarti2015incidence}.
Most of these give suboptimal approximation guarantees. In
particular, Saha and Getoor~\cite{saha2009maximum} provide a $\frac 1
4$-approximation algorithm for $k$-cover in one pass using
$\tilde{O}(m)$ space. The same technique gives a $\Theta(\log m)$
approximation algorithm for set cover in $\Theta(\log m)$ passes,
using $\tilde{O}(m)$ space.  On the hardness side, interestingly,
Assadi et al.~\cite{assadi2016tight} show that there is no
$\alpha$-approximation one-pass streaming algorithm for set cover
using $o(nm/\alpha)$ space.
%On the hardness side, Nisan~\cite{nisan2002communication} studies set
%cover in two-party communication setting and shows that computing a
%$(\frac 1 2 -\eps)\log_2m$ approximation to set cover requires
%$\Omega(m)$ randomized communication, for any constant $\eps$. This
%means that there is no $(\frac 1 2 -\eps)\log_2m$ approximation
%streaming algorithm for set cover using ${o}(m)$ space. 
Demaine et al.~\cite{demaine2014streaming} provide (for any positive
integer $r$) a $4^r \log m$-approximation algorithm for the set cover
problem in $4^r$ passes using $\tilde{O}(nm^{1/r}+m)$
space\footnotemark.  Recently, Har-Peled et al.\ improves this result
and provide a $p$-pass $O(p\log m)$-approximation algorithm in
$\tilde{O}(nm^{O(1/p)}+m)$ space\footnotemark[\thefootnote].  Indeed, all the
above results hold only for the set-arrival model.
%, and apart from the
%hardness results, they do not apply to the edge arrival model.

\footnotetext{
The space bounds claimed in~\cite{demaine2014streaming,har2016towards} assume $m = O(n)$, 
hence stated differently.\label{foot:nm-different}}

Often in the graph streaming problems, while the size of the input is
$\tilde{O}(|E|)$ for a graph $G(V,E)$, the solution size may be
as large as ${\Omega}(|V|)$. The best hope then is to find the
solution in $\tilde{O}(|V|)$ space. 
Algorithms fitting this description are called \emph{semi-streaming}~\cite{muthukrishnan2005data},
and many graph problems have been studied in this setting%
~\cite{ahn2009,AG11,EpsteinLMS09,mcgregor2005,feigenbaum2005graph,KelnerL11,KonradMM12,KonradR13}.
On the other hand, the extensive work on edge-arrival streaming%
~\cite{ahn2013spectral,andoni2014towards,assadi2015tight,chitnis2015kernelization,esfandiari2015streaming,kapralov2014approximating,kapralov2015streaming}
had not (prior to our owrk) studied coverage problems.

\subsection{Results and techniques}

%% First of all, in order to show the significance of our sketching technique, we rule
%% out the possibility of developing approximation algorithms
%% for coverage problems via accessing a $(1 \pm \eps)$-approximate oracle as a black
%% box. This hardness result shows the need for a new approach to solve
%% these problems. Such a hardness result has been observed for general
%% submodular functions~\cite{HS15}---and not for coverage functions---and
%% is of independent interest; see Section~\ref{sec:intro:epsilon}. 
%% % In particular, we show that having access to such an $\eps$-noise oracle is not useful in solving $k$-cover within any approximation factor better than $n^{0.49}$. 

\begin{table*}
\begin{center}
  \begin{tabular}{|c|c|c|c|c|c|} 
    \hline
    Problem  & Credit & \# passes & Approximation & Space & Arrival \\
    \hline
    \hline 
    $k$-cover & \cite{saha2009maximum} & $1$ & $ 1/ 4 $ & $\tilde{O}(m)$ & set \\ 
    $k$-cover & \cite{karbasiKDD2014} & $1$ & $ 1/ 2 $ & $\tilde{O}(n+m)$ & set \\ 
    $k$-cover & Here & 1 & $1- 1/ e -\eps$ & $\tilde{O}(n)$ & edge \\ 
    \hline
    Set cover w. outliers & \cite{emek2014semi,
      chakrabarti2015incidence} & $p$ &
    $O(\min(n^{\frac{1}{p+1}},e^{-\frac 1 p}))$ & $\tilde{O}(m)$ & set \\
		
    Set cover w. outliers & Here & $1$ & $(1+\eps)\log \frac 1 \lambda$ & $\tilde{O}_{\lambda}(n)$ & edge \\
    \hline
    Set cover & \cite{chakrabarti2015incidence,saha2009maximum} & $p$
    & $(p+1)m^{\frac 1 {p+1}}$ & $\tilde{O}(m)$ & set \\
		
    Set cover & \cite{demaine2014streaming} & $4^r$ & $4^r\log m$ &
    $\tilde{O}(nm^{\frac 1 r}+m)$ & set \\		
    Set cover  & \cite{har2016towards} & $p$ & $O(p\log m)$ &
    $\tilde{O}(nm^{O(\frac 1 p)}+m)$ & set \\		
    Set cover & Here & $p$ & $(1+\eps)\log m$ & $\tilde{O}(nm^{O(\frac
      1 p)}+m)$ & edge \\
    \hline
  \end{tabular}
\end{center}		
\caption{Comparison of results in streaming models. Note that all our results for edge arrival model also hold
	for the set arrival model.}
\label{tab:streaming}
\end{table*}

As our main result, we address the aforementioned shortcomings of existing algorithms for
coverage problems.
%, and develop almost optimal streaming algorithms for all the three
%coverage problems.
These
results are summarized in Table~\ref{tab:streaming}.
This paper is {\em the first to study the problem
in the edge-arrival model}, and present tight results for these
problems.
%% We accompany these results by matching lower bounds of the
%% $k$-cover problem. Finally, we demonstrate the flexibility and power
%% of the edge-arrival model by providing the first nontrivial result for
%% the dominating set problem (as a direct application). 
%% %
%% Our results are based on a novel sketching technique that we design for coverage functions. This sketching technique
%% can be applied in a very general setting to develop new scalable
%% algorithms for coverage problems. In particular, it enables us to
%% convert any $\alpha$-approximation algorithm for a large number of
%% coverage problems to a $(1-\eps)\alpha$-approximation
%% algorithm for the same problem in the streaming setting. See Sections~\ref{sec:intro:technique} and~\ref{sec:sketch}. Finally, we show the power of this technique by applying it to RAM and distributed optimization models.

\iffalse
\begin{table*}
\begin{center}
	\begin{tabular}{|c|c|c|c|c|} 		
		\hline
		Problem  & Credit & Approximation & Running time & Comment \\ 
		\hline 
		\hline 
		 $k$-cover & \cite{badanidiyuru2014fast,mirzasoleiman2014lazier} & $1- \frac 1
                 e -\eps$ & $\tilde{O}(nm)$ & submodular functions  \\ 
		$k$-cover & Here & $1- \frac 1  e -\eps$ & $\tilde{O}(n)$ &  - \\ 
		\hline	
	\end{tabular}
\end{center}
	\caption{Results for the RAM model.}
	\label{tab:RAM}
\end{table*}
\fi

\subsubsection{Streaming results}
%results for streaming and RAM models}
  \label{sec:intro:technique}
We present almost tight streaming algorithms for coverage problems. The following theorem states our main results formally.

\begin{theorem}\label{intro:streaming:alg}
  In the edge-arrival streaming model, for any arbitrary $\eps\in (0,1]$, there exist
% with
%  probability $1-\frac 1 n$ we have
  \begin{itemize} 
  \item (See Thm~\ref{thm:alg:kcover}) a single-pass
    $(1-\frac 1 e -\eps)$-approximation algorithm for $k$-cover using
    $\tilde{O}(n)$ space;
  \item (See Thm~\ref{thm:alg:epsSetCover}) a single-pass $(1+\eps)\log\frac 1 {\lambda}$-approximation
    algorithm for set cover with $\lambda$ outliers using
    $\tilde{O_{\lambda}}(n)$
    space; and
  \item (See Thm~\ref{thm:alg:SetCover}) a $p$-pass
    $(1+\eps)\log m$-approximation algorithm for set cover using
    $\tilde{O}(nm^{O(\frac{1}{p})}+m)$ space. 
\end{itemize}
\end{theorem}

%The above results are the first results presented for the edge arrival
%streaming model for coverage problems.
The above are the first such results for coverage problems in the streaming edge-arrival model. Moreover, they improve the approximation factor of previously
known results for the set-arrival model~\cite{saha2009maximum,nisan2002communication,emek2014semi,chakrabarti2015incidence}.  (However, in certain cases, the space complexities may be incomparable,
say, $\tilde{O}(n)$ versus $\tilde{O}(m)$.\footnote{Indeed, either $m$ or $n$ may be larger in practice~\cite{cormode2003comparing}. See also Footnotes~\ref{foot:nm-regime} and~\ref{foot:nm-different}.}) In fact, our result for
streaming set cover gives an exponential improvement over Demaine et
al.~\cite{demaine2014streaming} on both approximation factor and
number of rounds given the same space.  See Table~\ref{tab:streaming}
for comparison to previous work.
%\xxx {Please check this.}
Recently, Har-Peled et al.\ (Theorem~2.6 in~\cite{har2016towards})
provide a $p$-pass $O(p\log m)$-approximation algorithm in
$\tilde{O}(nm^{O(1/p)}+m)$ space in the set-arrival model. Notice that
our results for streaming set cover provide a better approximation
factor---i.e., $(1+\eps)\log m$ versus $O(p\log m)$---in the same space
and number of passes, while handling the more general edge-arrival
model.

On the hardness side, we show that any $\frac{1}{2}+\eps$-approximation
streaming algorithm for $k$-cover requires $\Omega(n)$ space. This
holds even for streaming algorithms with several passes.
%Interestingly, this hardness result holds for a very simple input in which $k=1$ and $m=2$. Therefore, there is no hope to find a $\frac{1}{2}+\eps$-approximation streaming algorithm for $k$-cover using $o(n)\Phi(m,k)$ space where $\Phi(m,k)$ is any arbitrary function of $m$ and $k$ but independent of $n$.

\begin{theorem}\label{thm:str:hard}
  Any $\frac{1}{2}+\eps$-approximation multi-pass streaming algorithm for
  $k$-cover requires $\Omega(n)$ space in total.
\end{theorem}

%%%%%%%%%%%%%%%%%%%%%%%%%%%%%%%%%%%%%%%

In a simultaneous and independent work, McGregor and Vu~\cite{mcgregor2016better} 
present a single-pass $1-1/e-\epsilon$ approximation algorithm for the $k$-cover
problem in the streaming setting with $\tilde{O}(n)$ space, using a different approach:
They directly analyze the behavior of the greedy algorithm on a specific noisy sketch, while we provide a sketch that translates any $\alpha$-approximation algorithm for $k$-cover to an
$(\alpha-\epsilon)$-approximation streaming algorithm using $\tilde{O}(n)$ space.   

\subsubsection{Sketching technique}
The main technique at the heart of our results
is a powerful sketching to summarize coverage
functions. As its main property, we show that any $\alpha$-approximate
solution to $k$-cover on this sketch is an $(\alpha
-\eps)$-approximate solution to $k$-cover on the original input with
high probability; see Theorem~\ref{thm:str:main}. Interestingly, this
sketch requires only $\tilde{O}(n)$ space. Our sketch is
fairly similar to $\ell_0$
sketches~\cite{cormode2003comparing}, which are
essentially defined to estimate the value of coverage functions; see
Appendix~\ref{Apx:O(nk)} for a formal definition. Indeed, one may
maintain $n$ instances of the $\ell_0$ sketch, and estimate the value
of the coverage function of a single feasible solution of size $k$
with high probability. However, having $n \choose k$ different
choices for a solution of size $k$ leads to a huge blow-up on the
failure probability  of at least one such solution. In
Appendix~\ref{Apx:O(nk)}, we show a straightforward analysis to
approximate $k$-cover using $\ell_0$ sketches with $\tilde{O}(nk)$
space, which is quite larger than our sketch.

%\xxx{MHB: add some more info about the sketching idea and in particular mention the bit complexity challenge.}

All the algorithms presented here construct $\tilde{O}(1)$ independent
instances of the sketch and then solve the problem
without any other direct access to the input. The simplicity of our
sketch enables its efficient construction and fast implementation of
the resulting algorithms.  Interestingly, this technique provides
almost tight approximation guarantees.  We remark that all the
algorithms presented in this work have success probabilities $1-\frac
1 n$; i.e., they may fail to produce the claimed solution with
probability $\frac 1 n$.  For simplicity we do not repeat this
condition elsewhere.

%%%%%%%%%%%%%%%%%%%%%%%%%%%%%%%%%%%%%%%%%

\iffalse %% RAM
%To show the power of this new sketching technique,  we also show how to apply and adapt this sketching technique for RAM and distributed models. 
In addition to the streaming setting, we demonstrate how to use this sketch and implement our algorithm in RAM model and provide a fast algorithm in this setting; see Table~\ref{tab:RAM}.  
%Vahab: We can remove the following material about RAM model. 
While the size of the input graph may be as large as $\Omega(nm)$, we show
that we can query only $\tilde{O}(n)$ bits of the input to construct
our sketch, paving the way for a $(1-\frac 1
e-\eps)$-approximation algorithm in the classical RAM model
\cite{aho1974design} that runs in $\tilde{O}(n)$ time. This is now the
fastest $(1-\frac 1 e -\eps)$-approximation algorithm in classical RAM
model.
\begin{theorem}
	In the RAM model, given $\eps\in (0,1]$, there exists a $(1-\frac 1 e
	-\eps)$-approximation algorithm for $k$-cover, with running time
	$\tilde{O}(n)$. %, succeeding with probability $1-\frac 1 n$.
\end{theorem}
\fi

Finally, in an accompanied paper, we also show how to apply this to distributed models, and design scalable distributed algorithms for covering problems. There we also confirm the effectiveness of this algorithm empirically on real data sets~\cite{bateni2016distributed}.\footnote{We decided to remove this part of the paper due to space constraints, and focus on the streaming applications.}

\subsubsection{A $(1 \pm \eps)$-approximate oracle is not sufficient}\label{sec:intro:epsilon}
There are several  sampling or sketching techniques that can be used
to develop a $(1 \pm \eps)$-approximate oracle $\cov_\eps$ to the coverage
function. One might hope that a black-box access to such an
oracle could be used as a subroutine in developing approximation
algorithms with good approximation guarantees. Here, we show that
this is not possible.
\begin{theorem}\label{thm:mainhard}
	Any $\alpha$-approximation algorithm for $k$-cover via oracle
	$\cov_{\eps}$ requires $\exp\left(\Omega({ n
		\eps^2\alpha^2}-{\log n} )\right)$ queries to the oracle. 
\end{theorem}
In particular, for any constant $\eps>0$, there is no polynomial-time
$n^{-0.49}$ approximation algorithm for $k$-cover given a
$(1 \pm \eps)$-approximate oracle $\cov_\eps$. This improves upon 
a similar hardness result for submodular functions~\cite{HS15}---and not 
for coverage functions. Our proof technique here
might be of independent interest.  (See details in Appendix~\ref{sec:epsError}.)

In order to prove Theorem~\ref{thm:mainhard}, first we define a
problem called \emph{$k$-purification} for which we show that any
randomized algorithm requires $\delta\exp\big(\Omega(\frac{
	\eps^2k^2}{n} )\big)$ oracle queries to succeed with probability
$\delta$. In a $k$-purification problem instance, we are given a
random permutation of $n$ items, with $k$ gold and
$n-k$ brass items. The types of individual items are not known to us.
We merely have access to an oracle
$\Pure_{\eps}(S)$ for $S\subseteq[1,n]$ defined as
\begin{align*}
\begin{cases}
0 & \text{if $\frac{k|S|}{n} - \eps \left(\frac{k|S|}{n} + \frac{k^2}{n}\right) \leq \Gold(S) \leq \frac{k|S|}{n} + \eps \left(\frac{k|S|}{n} + \frac{k^2}{n}\right) $,} \\
1 & \text{otherwise},
\end{cases}
\end{align*}
%with the following property: for any subset $S\subseteq
%[1,n]$, 
%\begin{itemize}
%	\item $Pure_{\eps}(S)=0$ if  $\frac{k|S|}{n} - \eps (\frac{k|S|}{n} + \frac{k^2}{n}) \leq Gold(S) \leq \frac{k|S|}{n} + \eps (\frac{k|S|}{n} + \frac{k^2}{n}) $ and
%	\item $Pure_{\eps}(S)=1$ otherwise
%\end{itemize}
where $\Gold(S)$ is the number of gold items in $S$. The goal in this problem is to find a set $S$ such that $\Pure_{\eps}(S)=1$.
The hardness proof is then based on a reduction between
$k$-purification and $k$-cover.

\subsection{Organization}
We next present the core idea behind our sketching technique and then
explain our algorithms in Section~\ref{sec:streaming}.  Due to
space constraints, most proofs and discussions appear in the
appendix.  In particular, we present in the appendix
our negative result for the black-box usage of $(1 \pm \eps)$-approximate oracles.

%We next discuss our impossibility result for black-box usage of the
%$\eps$-noise oracle to solve solve $k$ cover. Then we introduce our
%sketching technique in Section~\ref{sec:sketch}. The final section
%uses this sketch to first present algorithms for the three coverage
%problems under consideration. It eventually discusses each of the
%three computation frameworks and explains how the different algorithms
%and the sketching procedure can be implemented in each.

%Due to space constraints several proofs and discussions were moved to
%the appendix. 

%\section{Preliminaries}
%\input{prelims}

\section{Sketching for coverage problems}\label{sec:sketch}
%%%TODO
%\rm
\label{sec:sketch}

In this section we present a sketch $H_{\leq n}$ to approximate $k$-cover. Specifically, we show that any $\alpha$-approximate solution to $k$-cover on $H_{\leq n}$ is an $\alpha-O(\eps)$-approximate solution on the input graph, with high probability (see Theorem~\ref{thm:str:main}). Crucially $H_{\leq n}$ uses only $\tilde{O}(n)$ space. In order to define and prove the properties of $H_{\leq n}$, we introduce two intermediary sketches $H_p$ and $H'_p$, where $p \in [0, 1]$ is a parameter to be fixed later on.

In this section, we define the sketch in mathematical terms and establish its desirable properties.
Then in the following section, we discuss the intricacies of building and using it in the streaming model.

Let $h$ be a hash function mapping elements \elements to real numbers in $[0, 1]$.  First we throw away from the bipartite graph $G$ any element whose hash value exceeds $p$. This constructs $H_p$. In Lemma~\ref{lm:str:all} we show that, for sufficiently large $p$, any $\alpha$-approximate solution to $k$-cover on $H_p$ is an $\alpha-O(\eps)$-approximate solution on $G$, with high probability. Unfortunately, the number of edges in $H_p$ may be $\Omega (nk)$. 

Next we enforce an upper bound (defined below in terms of $n, k, \eps$) 
on the degree of elements in $H_p$, by
arbitrarily removing edges as necessary. This constructs $H'_p$.
Again for a sufficiently large choice of $p$, any $\alpha$-approximate
solution to $k$-cover on $H'_p$ is an $\alpha-O(\eps)$-approximate
solution on $G$, with high probability. Interestingly, if we select
$p$ wisely, $H'_p$ requires only $\tilde{O} (n)$ space. However, this
$p$ depends on the value of the optimum solution and may not be
accessible to the algorithm while constructing the sketch. To resolve
this issue, we define $H_{\leq n}$ with a similar structure as $H'_p$,
such that it always has $\tilde{O}(n)$ edges (see Definition~\ref{def:mainH}).
We remark that this conceptual description can be turned into
efficient implementations in several computational frameworks.
%, and
%later on we explain how to do this in streaming and RAM models.
%
Next comes the formal definitions of our sketch.

\begin{figure}
	\begin{center}
		\includegraphics[width=14cm]{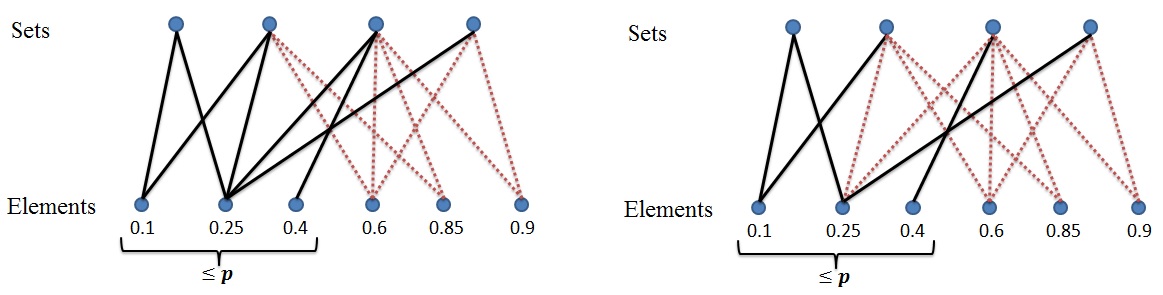}
	\end{center}
	\caption{Left figure is an example of $H_p$ and the right figure is an example of $H'_p$. In both figures, we have $p=0.5$. The number below each vertex is its hashed value. Solid edges are those included in the sketch and dotted edges are the rest of the edges in the graph.}
	\label{fig:H_p&H'_p}
\end{figure}

Let us overload the notation $h(.)$ such that $h(e)$, for an edge $e$,
denotes the value of $h$ on the endpoint of $e$ in \elements. 
For a fixed parameter $p$, we define $H_p$ to be the subgraph of $G$
induced by all vertices in \sets and the vertices in \elements with
$h$ less than $p$. In other words, $H_p$ contains an edge $e$ if and
only if $h(e)\leq p$. 
Let $H'_p$ be a maximal subgraph of $H_p$ such that the degree of the vertices of $H'_p$ in part \elements is at most $\frac{n\log(1/\eps)}{\eps k}$; as necessary we throw away edges arbitrarily. 
Below we define $H_{\leq n}(k,\eps,\delta'')$ based on $H'_p$. The former  is the
sketch used in all our algorithms. 
% % % % % % % % % % % % % % % % % % % % %

\begin{definition}\label{def:mainH}
	 For simplicity of
	 notation, we set $\delta = \delta'' \log\log_{1-\eps} m$.
Let $p^*$ be the smallest value such that the number of edges in
$H'_{p^*}$ is at least $\frac{24 n \delta\log(1/\eps) \log
	n}{(1-\eps)\eps^3}$.  Notice that $p^*$ is a function of the
randomness in the hash function. Remark that the number of edges in
$H'_{p^*}$ is at most $\frac{24 n\delta\log(1/\eps)
	\log n}{(1-\eps)\eps^3}+n\in \tilde{O}(n)$. We denote
$H'_{p^*}$ by $H_{\leq n}(k,\eps,\delta'')$, and drop the parameters from $H_{\leq n}(k,\eps,\delta'')$ when it is clear from the context. See Algorithm \ref{Alg:Hn}.

\end{definition}

\begin{algorithm*}%[!h]
	\textbf{Input:} An input graph $G$, $k$, $\eps\in (0,1]$, and $\delta''$.\\
	\textbf{Output:} Sketch $H_{\leq n}(k,\eps,\delta'')$.
	
	\begin{algorithmic}[1]
		\STATE Set $\delta = \delta''{\log\log_{1-\eps} m}$.
		\STATE Let $h$ be an arbitrary hash function that uniformly and
		independently maps \elements in $G$ to $[0,1]$.
		\STATE Initialize $H_{\leq n}(k,\eps,\delta'')$ with vertices
		\sets of $G$, and no edge.
		\WHILE {number of edges in $H_{\leq n}(k,\eps,\delta'')$ is less
			than $\frac{24 n\delta\log(1/\eps)\log n}{(1-\eps)\eps^3}$}
		\STATE Pick $v \in \elements$ of minimum $h(v)$
		that is still not in $H_{\leq n}(k,\eps,\delta'')$.
		\IF {degree of $v$ in $G$ is less than $\frac{n\log(1/\eps)}{\eps
				k}$} 
		\STATE Add $v$ along with all its edges to $H_{\leq
			n}(k,\eps,\delta'')$.
		\ELSE
		\STATE Add $v$ along with
		$\frac{n\log(1/\eps)}{\eps k}$ of its edges, chosen arbitrary, to
		$H_{\leq n}(k,\eps,\delta'')$. 
		\ENDIF
		\ENDWHILE
		%$H_{\leq n}(k,\eps,\delta'')$
		%
	\end{algorithmic}	
	\caption{$H_{\leq n}(k,\eps,\delta'')$}%{Algorithm \ref{PF:Alg1}}
	\label{Alg:Hn}
\end{algorithm*}

% % % % % % % % % % % % % % % % % % % % %

%Lemma \ref{lm:str:one} says that if $p$ is large enough, for a fixed
%set $S\subseteq \sets$, with high probability, $|\frac 1 p \Gamma(H_p,S) -
%\cov(S)| \leq \epsilon \opt_k$.

We argue that, for sufficiently large $p$, the quantity $\frac{1}{p}|\Gamma(H_p, S)|$ is a good estimate for
$\cov(S)$.  This is formalized below.

\begin{lemma}\label{lm:str:one}
  Pick $\frac{6\delta'}{\epsilon^2 \opt_k} \leq p \leq 1$, and let $S$
  be an arbitrarily subset of \sets such that $|S|\leq k$. With
  probability $1- e^{-\delta'}$ we have
\begin{align}
\left|\frac 1 p |\Gamma(H_p,S)| - \cov(S)\right| \leq \eps \opt_k.
\end{align}
\end{lemma}

%%%%%%%%%%%%%%%%%%%%%%%%%%%%%%%%%%%%%%%

In the following lemma we relate the approximate solutions on $H_p$ and $G$.

\begin{lemma}\label{lm:str:all}
  Pick $\frac{6 k\delta\log n}{\eps^2 \opt_k} \leq p \leq 1$. All
  $\alpha$-approximate solutions on $H_p$ are
  $(\alpha-2\eps)$-approximate solutions to the $k$-cover problem on
  $G$ with probability $1-e^{-\delta}$.
%Consequently, 
  Simultaneously for any set $S\subseteq \sets$ such that $|S|=k$, we have $\Big|\frac 1 p |\Gamma(H_p,S)| - \cov(S)\Big| \leq \eps \opt_k$.
\end{lemma}

%%%%%%%%%%%%%%%%%%%%%%%%%%%%%%%%%%%%%%%%%%%%%%%%%%%%%%

%
The following lemma relates the solutions on $H'_p$ and $H_p$.

\begin{lemma}\label{lm:str:H'good}
  Pick arbitrary $0 \leq p\leq 1$ and $1\leq k\leq n$. Any
  $\alpha$-approximate solution of $k$-cover on $H'_p$ is an
  $\alpha(1-\eps)$-approximate solution on $H_p$. 
  %In particular, $0\leq \opt_k(H'_p)-\opt_k(H_p)\leq \eps \opt_k(H_p)$.
\end{lemma}
\begin{proof}
  Let $\opt_{H}$ and $\opt_{H'}$ be subsets of \sets with size $k$
  that maximize $|\Gamma(H_p,\opt_{H})|$ and
  $|\Gamma(H'_p,\opt_{H'})|$, respectively. Remark that $H'_p$ is a
  subgraph of $H_p$, hence $|\Gamma(H'_p,S)|\leq
  |\Gamma(H_p,S)|$ for any $S\subseteq \sets$. Later we show that
  there exists a set $R$ of size $k$ such that $|\Gamma(H'_p,R)| \geq
  (1-\eps)|\Gamma(H_p,\opt_{H})|$. Thus, for an $\alpha$-approximate
  solution $S$ on $H'_p$, we have
\begin{align*}
|\Gamma(H_p,S)|
&\geq  |\Gamma(H'_p,S)| 
&\text{since $H'_p \subseteq H_p$,}
\\
&\geq \alpha |\Gamma(H'_p,\opt_{H'})|
&\text{as $S$ is an $\alpha$-approximate solution,}
\\
&\geq \alpha |\Gamma(H'_p,R)| 
&\text{by definition of $\opt_{H'}$,}
\\
&\geq \alpha (1-\eps)|\Gamma(H_p,\opt_{H})|.
%\\
%&\geq (\alpha-\eps)|\Gamma(H_p,\opt_{H})|.
\end{align*}
%Therefore, $S$ is an $(\alpha-\eps)$-approximate solution on $H_p$, as
%desired. 

To prove the existence of a suitable $R$, we follow a probabilistic
argument, producing a randomized set $R^*$ of size $k$ such that 
%Now, we just need to show that there exists a set $R$ such
%that $|\Gamma(H'_p,R)| \geq (1-\eps)|\Gamma(H_p,\opt_{H})|$. We prove
%the existence of a random subset $R^*$ of size $k$ such that
$\E[|\Gamma(H'_p,R^*)|] \geq (1-\eps)|\Gamma(H_p,\opt_{H})|$.
%Then
%guaranteed existence of a set $R$ with $|\Gamma(H'_p,R)| \geq
%\E[|\Gamma(H'_p,R^*)|]$ finishes the argument.

We construct $R^*$ by removing $\eps k$ sets from $\opt_H$ uniformly
at random, and adding $\eps k$ sets from \sets uniformly at
random. Note that each element in $\Gamma(H_p,\opt_{H})$  with degree
at most $\frac{n\log(1/\eps)}{\eps k}$ in $H_p$ appears 
in $\Gamma(H_p,R^*)$ with probability $1-\eps$, hence in
$\Gamma(H'_p,R^*)$. Now let us consider a high-degree element
$u$---one with degree at least $\frac{n\log(1/\eps)}{\eps k}$ in
$H_p$, i.e., degree exactly $\frac{n\log(1/\eps)}{\eps k}$ in
$H'_p$. The probability that $u$ is not included in any of the $\eps
k$ randomly added sets is at most
\begin{align*}
\left(1-\frac{\frac{n\log(1/\eps)}{\eps k}}{n}\right)^{\eps k}
=\left(1-\frac{\log(1/\eps)}{\eps k}\right)^{\eps k}
=\left(1-\frac{\log(1/\eps)}{\eps k}\right)^{\frac{\eps k}{\log(1/\eps)}{\log \frac 1 {\eps}}}
\leq \left(\frac 1 e\right)^{\log\frac{1}{\eps}}
= \eps.
\end{align*}
Therefore, each vertex in $\Gamma(H_p,\opt_{H})$ exists in
$\Gamma(H'_p,R^*)$ with probability at least $1-\eps$, proving the
claim $\E[|\Gamma(H'_p,R^*)|] \geq (1-\eps)|\Gamma(H_p,\opt_{H})|$.
\end{proof}

%%%%%%%%%%%%%%%%%%%%%%%%%%%%%%%%%%%%%%%%%%%%%%%%%%%

In the following two lemmas, we argue that maintaining the solution in the
reduced-degree subgraph $H'_p$ does not require too much memory.

\begin{lemma} \label{lm:str:CpN()<} Pick arbitrary $C\geq 1$ and let
  $p= \frac{6 Ck\delta\log n}{\eps^2 \opt_k}$. With probability at
  least $1-e^{1-\delta}$, we have
\begin{align*}
 \max_{S\subseteq \sets:|S|=k}|\Gamma(H'_p,S)|\leq \frac{12Ck\delta\log n}{\eps^2}.
\end{align*}
\end{lemma}

%\xxx{Done: what's the purpose of the following lemma?}
%The following lemma relates the value of the optimum solution of
%$H'_p$ to the number of edges of $H'_p$.

\begin{lemma}\label{lm:str:N()>}
  Pick arbitrary $0\leq p\leq 1$ and $1\leq k\leq n$, and let $m'_p$
  denote the number of edges in $H'_p$. We have
\begin{align*}
m'_p \frac{\eps k}{2n\log(1/\eps)} \leq |\Gamma(H'_p,\opt_{H'})|.\end{align*}
\end{lemma}

%%%%%%%%%%%%%%%%%%%%%%%%%%%%%%%%%%%%%%%%%%%%%%%
%%%%%%%%%%%%%%%%%%%%%%%%%%%%%%%%%%%%%%%%%%%%%%%
%%%%%%%%%%%%%%%%%%%%%%%%%%%%%%%%%%%%%%%%%%%%%%%

The following theorem relates the approximate solutions on $H_{\leq n}$ and $G$. 

\begin{theorem}\label{thm:str:main}
	Let $\delta''\in [1,\infty)$ and $k\in[1,n]$ be two arbitrary
        numbers. Any $\alpha$-approximate solution to $k$-cover on
        $H_{\leq n}$ is an $\alpha-12\epsilon$ approximation solution
        on $G$, with probability $1-3e^{-\delta''}$. 
\end{theorem}

%\xxx{Done: more explanation and intuition needed.}
The following lemma provides a bicriteria bound on the coverage of
solutions in
$H_{\leq n}$, where $k'\leq k$ is the size of the set cover on $G$.
This lemma will be useful in obtaining results for set cover and set
cover with outliers.
\xxx{Check the wording of the following.}

\begin{lemma}\label{lm:loglambda}
  Let $k'$ be the size of the minimum set cover on the input graph
  $G$, and let $k=\xi k'$. There exists a solution of size $k'$ on
  $H_{\leq n}(k,\eps,1)$ that covers at least $1-\xi\eps$ fraction of the
  elements in $H_{\leq n}(k,\eps,1)$.
\end{lemma}

\section{The streaming setting}\label{sec:streaming}
Indeed, with no time constraint, one can use
$\ell_0$ sketches and give a $1-\eps$ approximation streaming algorithm
for $k$-cover in $\tilde{O}(nk)$ space; see Appendix~\ref{Apx:O(nk)}.
This simple streaming algorithm constructs a $(1 \pm \eps)$-approximate oracle to
the value of the coverage function, using $\tilde{O}(nk)$ space. One
can use this algorithm and try  all solutions of size
$k$ to find a $1-2\eps$ approximate solution of $k$-cover. However, as
Theorem~\ref{thm:mainhard} states, using this oracle and without any
further assumptions, there is no polynomial time
$n^{-0.5+\eps}$-approximation algorithm for $k$-cover. In addition, the
space used by this algorithm may be quite large for large values of
$k$.

In this section, we improve the algorithm provided in
%% TODO: changed from $1-\frac 1 e \eps$ to $1-\frac 1 e - \eps$
Appendix~\ref{Apx:O(nk)} and give a $1-\frac 1 e - \eps$-approximation
one-pass streaming algorithm for $k$-cover, using $\tilde{O}(n)$
space. This is done by first constructing $H_{\leq n}$ in the streaming
setting and then providing efficient algorithms that only access the sketch $H_{\leq n}$.
Using the same technique, we give a $(1+\eps)\log\frac 1 {\lambda}$
approximation one-pass streaming algorithm for set cover with
outliers, using $\tilde{O}_{\lambda}(n)$ space. Besides, for any
arbitrary $r\in [1,\log m]$, we give a $(1+\eps)\log m$ approximation
$r$-pass streaming algorithm for set cover, using
$\tilde{O}(nm^{O(1/r)}+m)$ space. Interestingly, the update times of
all our algorithms are $\tilde{O}(1)$.

On the hardness side we show in Appendix~\ref{sec:strHard} that any
$\frac 1 2 + \eps$ approximation streaming algorithm for the $k$-cover
problem requires ${\Omega}(n)$ space. This rules out the
existence of $\frac 1 2 + \eps$ approximation parametrized streaming
algorithms for the $k$-cover problem and shows that the space of our
algorithm is tight up to a logarithmic factor.

Next we show how to construct $H_{\leq n}$ in the streaming setting. Note that to define $H_{\leq n}$ we map (via a hash function) each
element to a number in $[0,1]$ independently. Such a mapping requires
$\tilde{O}(m)$ random bits. However, we use a simple equivalent random
process to construct $H_{\leq n}$ using $\tilde{O}(|H_{\leq n}|)$
random bits, where $|H_{\leq n}|$ is the number of edges in $H_{\leq
  n}$.

Let $p^*$ be the probability corresponding to $H_{\leq n}$. Remark
that the number of elements $v$ with $h(v) \leq p^*$ is at most
$|H_{\leq n}|$, i.e., at most equal to the number of edges in $H_{\leq
  n}$. Note that, if we know that the hash value of an element is
greater than $p^*$, we can simply remove that element. Thus, at the
beginning we iteratively sample $|H_{\leq n}|$ elements without
replacement, and assume that this sequence is indeed that of the first
$|H_{\leq n}|$ elements ordered by their hashed value. This process
requires only $\tilde{O}(|H_{\leq
  n}|)$ random bits.

\begin{algorithm*}%[!h]
	\textbf{Input:} An input graph $G$, $k$, $\eps\in (0,1]$, and $\delta''$.\\
	\textbf{Output:} Sketch $H_{\leq n}(k,\eps,\delta'')$.\\
	\textbf{Initialization:}
	\begin{algorithmic}[1]
		\STATE Set $\delta = \delta'' {\log\log_{1-\eps} m}$.
		\STATE Pick $\frac{24 n\delta\log(1/\eps)\log n}{(1-\eps)\eps^3}+ \frac{n\log(1/\eps)}{\eps k}$ element from \elements uniformly at random and let $\Pi$ be a random permutation over these elements.
		\STATE Initialize $H_{\leq n}(k,\eps,\delta'')$ with vertices
		\sets of $G$, and no edge.
	\end{algorithmic}	
	\textbf{Update edge $(u,v)$:}
	\begin{algorithmic}[1]
		\IF {$v$ is not sampled in $\Pi$}
			\STATE Discard $(u,v)$.
		\ELSIF {degree of $v$ in $G$ is $\frac{n\log(1/\eps)}{\eps k}$}
			\STATE Discard $(u,v)$.
		\ELSE
		\STATE Add $(u,v)$ to $H_{\leq
			n}(k,\eps,\delta'')$.
		\ENDIF
		\WHILE {number of edges in $H_{\leq n}(k,\eps,\delta'')$ is more
			than $\frac{24 n\delta\log(1/\eps)\log n}{(1-\eps)\eps^3}+ \frac{n\log(1/\eps)}{\eps k}$}
		\STATE Let $w$ be the last element in $\Pi$.
		\STATE Remove $w$ from $\Pi$.
		\STATE Remove $w$ from  $H_{\leq n}(k,\eps,\delta'')$.
		\ENDWHILE
		%$H_{\leq n}(k,\eps,\delta'')$
		%
	\end{algorithmic}	
	\caption{Streaming algorithm to compute $H_{\leq n}(k,\eps,\delta'')$}
	\label{Alg:str:Hn}
\end{algorithm*}

%\section{Algorithms}\label{sec:algorithms}

Next we 
describe how to use the sketch to solve each of the three problems:
$k$-cover, set cover, and set cover with outliers. As a result, we provide tight and almost tight
streaming algorithms for $k$-cover, set cover, and set cover with outliers.

The greedy algorithm for $k$-cover iteratively selects a vertex that
increases the valuation function $f$ the most and adds it to the
solution. Let $\greedy(k,G)$ denote the set of $k$ vertices picked by
the greedy algorithm when run on input graph $G$. It is known that the
\greedy is a $1-\frac 1 e$ approximation
algorithm~\cite{nemhauser1978analysis}. In addition, we know that
$\cov(\greedy(k\log\frac 1 {\lambda},G)) \leq (1-\lambda)\opt_k(G)$.

%% changed \vareps => \eps
%% changed \eps => eps'
\begin{algorithm*}%[!h]
  \textbf{Input:} An input graph $G$, $k$, and $\eps\in (0,1]$.\\
  \textbf{Output:} A $1-\frac 1 e -\eps$ approximate solution to
  $k$-cover on $G$ with probability $1-\frac 1 n$.
	
  \begin{algorithmic}[1]
    \STATE Set $\delta''=2+\log n$ and $\eps'=\frac 1 {12}\eps$. 
    \STATE Construct sketch $H_{\leq n}(k,\eps',\delta'')$. \hspace{1cm} $//$\emph{ Compute this over the stream.}
    \STATE Run the greedy algorithm (or any $1-\frac{1}{e}$
    approximation algorithm) on this
    sketch %$H_{\leq n}(k,\eps',\delta'')$
    and report $\greedy(k,H_{\leq n}(k,\eps',\delta''))$.
  \end{algorithmic}
	
  \caption{$k$-cover}%{Algorithm \ref{PF:Alg1}}
  \label{Alg:main:KCover}
\end{algorithm*}

\begin{theorem} \label{thm:alg:kcover} For any $\eps\in (0,1]$ and any
  graph $G$, Algorithm~\ref{Alg:main:KCover} produces a $(1-\frac 1 e
  -\eps)$-approximate solution to $k$-cover on $G$ with probability
  $1-\frac 1 n$. The number of edges in the sketch used by this
  algorithm is $\tilde{O}(n)$.
\end{theorem}

%%%%%%%%%%%%%%%%%%%%%%%%%%%%%%%%%%%%%%%%%%%%%%%%%%%%%%%%%%%%%

%%%%%%%%%%%%%%%%%%%%%%%%%%%%%%%%%%%%%%%%%%%%%%%%%%%%%%%%%%%%%
%%%%%%%%%%%%%%%%%%%%%%%%%%%%%%%%%%%%%%%%%%%%%%%%%%%%%%%%%%%%%

%%%%%%%%%%%%%%%%%%%%%%%%%%%%%%%%%%%%%%%%%%%%%%%%%%%%%%%%%%%%%

\begin{algorithm*}%[!h]
  \textbf{Input:} Parameters $k'$, $\eps'\in (0,1]$,
  $\lambda'\in (0,\frac 1 e]$, and $C'\in [1,\infty)$, as well as a
  graph $G$ promised to have a set cover of size $k'$.\\
  \textbf{Output:} A solution of size $k'\log\frac 1 {\lambda'}$
  covering $1-\lambda'-\eps'$ fraction of \elements in $G$ with
  probability $1-\frac 1 {C' n}$.
	
  \begin{algorithmic}[1]
    \STATE Set $\delta''= \log_{1+\eps}n[\log(C'n)+2]$ and
    $\eps=\frac {\eps'} {13 \log\frac 1 {\lambda'}}$.
    \STATE Construct sketch $H_{\leq n}(k'\log\frac 1
    {\lambda'},\eps,\delta'')$. \hspace{1cm} $//$\emph{ Compute this over the stream.}
    \STATE Run the greedy algorithm on this sketch to get solution $S=\greedy(k'\log\frac 1 {\lambda'},H_{\leq n})$
    \IF {$S$ covers at
      least $1-\lambda'- \eps\log\frac{1}{\lambda'}$ fraction of
      \elements in $H_{\leq n}$}
    \RETURN $S$
    \ELSE
    \RETURN false
    \ENDIF
  \end{algorithmic}
	
  \caption {A submodule to solve set cover}%{Algorithm \ref{PF:Alg1}}
  \label{Alg:main:epsSetCover0}
\end{algorithm*}

\begin{lemma}\label{lm:epsSC0}
  For arbitrary $k'$, $\eps'\in (0,1]$, $\lambda'\in (0,\frac 1 e]$,
  $C'\in [1,\infty)$, and graph $G$,
  Algorithm~\ref{Alg:main:epsSetCover0} returns false only if the size
  of the minimum set cover of $G$ is greater than $k'$. Otherwise, the
  algorithm returns a solution of size $k'\log\frac 1 {\lambda'}$
  that covers $1-\lambda'-\eps'$ fraction of \elements in $G$ with
  probability $1-\frac 1 {C' n}$. The number of edges in the sketch
  used by this algorithm is $O(\frac {n \log^2n\log^6m\log C'}{
    {\eps'} ^3})$.
%\begin{itemize}
%\item If $G$ contains a set cover of size $k'$, Algorithm \ref{Alg:main:epsSetCover0} returns a solution of size $\log(\frac 1 {\lambda'})k'$ that covers $1-\lambda'-\varepsilon'$ fraction of vertices on $G$.
%\item If Algorithm \ref{Alg:main:epsSetCover0} returns a solution it covers at least $1-\lambda'-\varepsilon'$ fraction of vertices on $G$.
%\item 
%\end{itemize}
\end{lemma}

%%%%%%%%%%%%%%%%%%%%%%%%%%%%%%%%%%%%%%%%%%%%%%%%%%%%%%%%%%%%%

%%%%%%%%%%%%%%%%%%%%%%%%%%%%%%%%%%%%%%%%%%%%%%%%%%%%%%%%%%%%%
%%%%%%%%%%%%%%%%%%%%%%%%%%%%%%%%%%%%%%%%%%%%%%%%%%%%%%%%%%%%%

%%%%%%%%%%%%%%%%%%%%%%%%%%%%%%%%%%%%%%%%%%%%%%%%%%%%%%%%%%%%%

\begin{algorithm*}%[!h]
  \textbf{Input:} A graph $G$ and parameters $\eps\in [0,1]$,
  $\lambda\in (0,\frac 1 e]$, and $C\geq 1$.\\
  \textbf{Output:} A $(1+\eps)\log\frac 1 {\lambda}$ approximate
  solution to set cover with $\lambda$ outliers on $G$ with
  probability $1-\frac 1 {Cn}$.
  
  \begin{algorithmic}[1]
    \STATE Set $\eps' = {\lambda}(1-e^{-\eps/2})$, and
    $\lambda'= \lambda e^{-\eps/2}$, and
    $C'=C\log_{1+\frac{\eps}{3}}n$, and $k'=1$.  
      \REPEAT
      \STATE $k' \gets (1+\frac{\eps}{3})k'$
      \STATE Run Algorithm \ref{Alg:main:epsSetCover0} on $(k',
      \eps',\lambda',C',G)$ and let $S$ be the outcome. \hspace{0.6cm} $//$ \emph{Run these in parallel.}% \\\hspace{11.55cm}$\backslash\backslash$ over the stream.
      \UNTIL{$S$ is not false or $k'=n$}

      \RETURN $S$
    \end{algorithmic}
	
    \caption {Set cover with $\lambda$
      outliers}%{Algorithm \ref{PF:Alg1}}
    \label{Alg:main:epsSetCover}
\end{algorithm*}

\begin{theorem} \label{thm:alg:epsSetCover} Given $\eps\in (0,1]$,
  $C\geq 1$ and a graph $G$, Algorithm~\ref{Alg:main:epsSetCover}
  returns a $(1+\eps)\log\frac 1 {\lambda}$ approximate
  solution to set cover with $\lambda$ outliers on $G$ with
  probability $1-\frac 1 n$.  The total number of edges in the
  sketches used by this algorithm is $\tilde{O}(n/\lambda^3)
  \subseteq\tilde{O}_{\lambda}(n)$.
\end{theorem}

%%%%%%%%%%%%%%%%%%%%%%%%%%%%%%%%%%%%%%%%%%%%%%%%%%%%%%%%%%%%%

%%%%%%%%%%%%%%%%%%%%%%%%%%%%%%%%%%%%%%%%%%%%%%%%%%%%%%%%%%%%%
%%%%%%%%%%%%%%%%%%%%%%%%%%%%%%%%%%%%%%%%%%%%%%%%%%%%%%%%%%%%%
	
%%%%%%%%%%%%%%%%%%%%%%%%%%%%%%%%%%%%%%%%%%%%%%%%%%%%%%%%%%%%%

\begin{algorithm*}%[!h]
  \textbf{Input:} A graph $G$ as well as $\eps\in (0,1]$, $C\geq 1$,
  and $r \in [1,\log m]$.
  \\
  \textbf{Output:} A $(1+\eps)\log m$ approximate solution to set
  cover of $G$ with probability $1-\frac 1 {Cn}$.
	
  \begin{algorithmic}[1]
    \STATE Let $G_1=G$, $\lambda =m^{-\frac{1}{2+r}}$ , $C'=(r-1)C$,
    $S=\emptyset$.
    \FOR{$i=1$ \TO $r-1$}
    \STATE Run Algorithm~\ref{Alg:main:epsSetCover} on
    $(G_i,\eps,\lambda,C')$ and let $S_i$ to be the outcome. \hspace{0.5cm} $//$ \emph{$i$-th streaming pass.}
    \STATE Add $S_i$ to $S$
    \STATE Remove from $G_i$ the elements covered by $S_i$ and call the new graph $G_{i+1}$.
    \ENDFOR
    \STATE Run the greedy algorithm to find a set cover of $G_r$ and let $S^{\greedy}$
    to be the result.
    \STATE Add $S^{\greedy}$ to $S$.
    \RETURN $S$
  \end{algorithmic}
	
  \caption {Set cover in $r$ iterations}
  \label{Alg:main:SetCover}
\end{algorithm*}

We implement each iteration of Algorithm~\ref{Alg:main:SetCover} in
two streaming passes. In the first pass of each iteration we simply
mark covered elements to virtually construct $G_i$, whereas in the
second pass, we construct $H_{\leq n}$. After all $r-1$ iterations, we
utilize one extra pass to keep all edges to construct $G_r$. Hence, the following theorem proves the third statement of Theorem~\ref{intro:streaming:alg}

\begin{theorem}\label{thm:alg:SetCover}
  Given $\eps\in (0,1]$ and a graph $G$,
  Algorithm~\ref{Alg:main:SetCover} finds a $(1+\eps)\log m$
  approximate solution to set cover on $G$ with probability $1-\frac
  1 n$. The total number of edges in the sketches used by this
  algorithm plus the number of edges in $G_r$ is at most
  $\tilde{O}(nm^{\frac{3}{2+r}}) \subseteq \tilde{O}(nm^{O(1/r)})$.
\end{theorem}

%%% TODO
\begin{proof}
  The algorithm runs $r-1$ instances of
  Algorithm~\ref{Alg:main:epsSetCover}. Theorem~\ref{thm:alg:epsSetCover}
  holds for each with probability $1- \frac 1 {C'n} = 1- \frac 1
  {(r-1)Cn}$, hence for all simultaneously with probability $1-
  (r-1)\frac 1 {(r-1)Cn} = 1- \frac 1 {Cn}$.  We assume these hold
  when proving the statement of the theorem.

  Let $k'$ be the size of the minimum set cover in $G$. Note that for
  any $i\in [1,r]$, $G_i$ is an induced subgraph of $G$ that contains
  all sets in $G$. Thus, any set cover of $G$ is a set cover of $G_i$
  as well. This means that the size of the set cover of $G_i$ is at most
  $k'$. Therefore, Theorem~\ref{thm:alg:epsSetCover} bounds the number of
  sets chosen by each run of Algorithm~\ref{Alg:main:epsSetCover} by
  $(1+\eps)\log\frac 1 {\lambda}k'= (1+\eps)\log
  m^{\frac{1}{2+r}}k'$. Also, each run of
  Algorithm~\ref{Alg:main:epsSetCover} covers $1-\lambda$ fraction of
  the remaining uncovered elements. Therefore, the number of uncovered
  elements in $G_i$ is at most $m\lambda^{i-1}$, and in particular this is $m\lambda^{r-1}=
  m^{\frac{3}{2+r}}$ for $G_r$. 
%Thus, the result is a $\log m^{\frac{3}{2+r}}$
%  approximate solution on $G_r$. 
  Therefore, the total size of the set cover obtained by this
  algorithm is at most
\begin{align*}
(r-1)(1+\eps)k'\log m^{\frac{1}{2+r}} + k'\log m^{\frac{3}{2+r}} &\leq 
(1+\eps) k' \left[(r-1)\log m^{\frac{1}{2+r}}+\log m^{\frac{3}{2+r}}\right]
\\& =
(1+\eps) k'\left[(r-1){\frac{1}{2+r}}+{\frac{3}{2+r}}\right]\log m
\\&=
(1+\eps)k'\log m.
\end{align*}

Remark that the total number of edges in the sketches used by
Algorithm~\ref{Alg:main:epsSetCover} is $\tilde{O}(n/\lambda^3)$. With
$r\leq \log m$ such runs, the total number of edges in all the
sketches is $\tilde{O}(n/\lambda^3) = \tilde{O}(n
m^{\frac{3}{2+r}})$. On the other hand, the number uncovered elements
in $G_{r}$ is $m\lambda^{r-1}= m^{\frac{3}{2+r}}$. Thus, the number of
edges in $G_{r}$ is at most $n m^{\frac{3}{2+r}}$. Therefore, the
total number of edges in the sketches plus the
number of edges in $G_r$ is $\tilde{O}(nm^{\frac{3}{2+r}})
\subseteq \tilde{O}(nm^{O(1/r)})$.
\end{proof}

\section{Conclusion}

In this paper, we presented a simple, yet powerful sketching technique for coverage problems, and showed how to construct this sketch in streaming model. The streaming results improve the state of the art in three dimensions: approximation ratio, space complexity, and streaming arrival model (i.e., from set-arrival to element- or edge-arrival model). 
 In an accompanied paper, we also applied this sketching idea for distributed computation models (such as MapReduce), and show how it improves the best known results in that area as well. More notably, we also performed an extensive empirical evaluation of resulting distributed algorithms and show the effectiveness of applying this sketching technique for analyzing massive data sets in practice~\cite{bateni2016distributed}. As noted earlier, this sketch and the distributed and streaming algorithms based on it work very well in instances in which the size of the subsets is large. Notably, all the other techniques (e.g., based on composable core-sets) fail in these regimes. As future research, we hope this technique can be applied to other computation models and other problems.

\bibliographystyle{plain}
\bibliography{k-cover}

\appendix

\section{The $k$-cover problem via $(1 \pm \eps)$-approximate oracle}\label{sec:epsError}

In this section we consider the approximability of $k$-cover using the
$(1 \pm \eps)$-approximate oracle, and prove Theorem~\ref{thm:mainhard} by showing
that any $\alpha$-approximation algorithm via
oracle $\cov_{\eps}$ requires at least $\exp\left(\Omega({ n
    \eps^2\alpha^2}-{\log n} )\right)$ oracle queries.

Theorem \ref{thm:purehard} states the hardness of the $k$-purification
problem.  Its proof uses the following generalization of the Chernoff
bound.

\begin{lemma} \label{lm:mychernoff} Let $X$ be the sum of several
  negatively correlated binary random variables.  We have
\begin{align*}
  \Pr\left( |X-\E[X]| > \gamma \right) \leq 2
  \exp\left(-\frac{\gamma^2}{3\E[X]} \right).
\end{align*}
\end{lemma}

\begin{proof}  %{Lemma~\ref{lm:mychernoff}}
  Panconesi and Srinivasan \cite{panconesi1997randomized} show that if
  $X$ is the sum of certain negatively correlated binary random
  variables, we have $\Pr\left( |X-\E[X]| > \eps \E[X] \right) \leq 2
  \exp\left(-\frac{\eps^2 \E[X]}{3} \right)$. Setting $\gamma = \eps
  \E[X]$ yields
$$\Pr\left( |X-\E[X]| > \gamma \right) \leq  2 \exp\left(-\frac{\left(\frac{\gamma}{\E[X]}\right)^2 \cdot\E[X]}{3} \right)=  2 \exp\left(-\frac{\gamma^2}{3\E[X]} \right)$$ as desired.
\end{proof}

%%%%%%%%%%%%%%%%%%%%%%%%%%%%%%%%%%%%%%%%%%%%%%%%%%%%%%%%%%%%%%%%%%%%%%

\begin{theorem}\label{thm:purehard}
  Any randomized algorithm that solves
  $k$-purification with probability at least $\delta$ requires at least
  $\delta\exp\left(\Omega(\frac{ \eps^2k^2}{n} )\right)$ oracle queries.
\end{theorem}
\begin{proof}
  By Yao's principle we can restrict our analysis to deterministic
  algorithms. Let \Alg be a deterministic algorithm for
  $k$-purification. Suppose that after $q$ queries, algorithm \Alg
  finds with probability $\delta$ a set $S$ such that
  $\Pure_{\eps}(S)=1$.  Let $S_1,S_2,\dots,S_q$ be the $q$ subsets
  queried by \Alg.  Definition of $\delta$ and the union bound give
\begin{align}\label{eq:q>sumpi}
\delta \leq \sum_{i=1}^{q} \Pr\big( \Pure_{\eps}(S_i)=1\big).
\end{align}

Now, we provide an upper bound to $\Pr\big( \Pure_{\eps}(S)=1\big)$ for an
arbitrary subset $S$.  Let $X_i$ be a random variable that indicates
whether the $i$-th item in $S$ is gold. Let $X=\sum_{i=1}^{|S|} X_i$.
Indeed, $X_i$ variables are negatively
correlated~\cite{joag1983negative}.  We set $\gamma=\eps
(\frac{k|S|}{n} + \frac{k^2}{n})$ in Lemma~\ref{lm:mychernoff} to
obtain
\begin{align*}
\Pr\big( \Pure_{\eps}(S)=1\big)&=\Pr\left( |X-\E[X]| > \eps \left(\frac{k|S|}{n} + \frac{k^2}{n}\right) \right) 
\\&\leq  2 \exp\left(-\frac{ \eps^2(\frac{k|S|}{n} + \frac{k^2}{n})^2}{3\E[X]} \right)
\\&\leq  2 \exp\left(-\frac{ \eps^2(\frac{k|S|}{n} + \frac{k^2}{n})^2}{3\frac{k|S|}{n}} \right)
\\&\leq  2 \exp\left(-\frac{ \eps^2k(|S|+ k)^2}{3n|S|} \right)
\\&\leq  2 \exp\left(-\frac{ \eps^2k^2}{3n} \right).
\end{align*}

Coupled with Inequality~\eqref{eq:q>sumpi} the above implies that
$\delta \leq \sum_{i=1}^{q} \Pr\big( \Pure_{\eps}(S_i)=1\big)\leq 2 q
\exp\left(-\frac{ \eps^2k^2}{3n} \right)$, which means $q\geq
\frac{\delta}{2} \exp\left(\frac{ \eps^2k^2}{3n} \right)$, as desired.
\end{proof}

%%%%%%%%%%%%%%%%%%%%%%%%%%%%%%%%%%%%%%%%%%%%%%%%%%%%%%%%%%%%%%%%%%%%%%

\begin{proofof}{Theorem \ref{thm:mainhard}}
  Given an instance of the $k$-purification problem we construct a
  $k$-cover instance with a $(1 \pm \eps')$-approximate oracle as follows.  We
  associate one set for each gold or brass item in the original
  instance in such a way that the value of the coverage function (for
  nonempty $S$) is $\cov(S)= k + \frac{n}{k}\Gold(S)$; i.e., there are
  $k$ elements common between all gold and brass \emph{sets}, and in
  addition, each gold set contains $\frac{n}{k}$ additional exclusive
  elements. The optimum solution consists of all gold sets, hence
\begin{align}\label{eq:err:opt>n}
\opt=k+\frac{n}{k}k = k+n > n.
\end{align}

We define 
$$  \cov_{\eps'}(S)= 
\begin{cases} 
k + |S| &\mbox{if } \Pure_{\eps}(S)=0 \\
\cov(S) &\mbox{otherwise.}
\end{cases}%
$$

%\begin{claim}\label{clm:delta-noise}
We claim that
$\cov_{\eps'}$ is a $(1 \pm \eps')$-approximate oracle to $f$ for  $\eps' = 2\eps$.
%\end{claim}
%%
%%
%\begin{proofof}{Claim~\ref{clm:delta-noise}}
%Remark that when $\Pure_{\eps}(S)=1$, we have $(1-\eps)\cov(s)\leq
%\cov_{\eps}(S)= \cov(S)\leq (1+\eps)\cov(S)$. 
We set $\eps' = 2\eps$. Notice that for $\Pure_{\eps}(S)=1$, the
estimate $\cov_{\eps'}(S)$ is clearly within the $1\pm\eps'$ factor of
$\cov(S)$.  Moreover when $\Pure_{\eps}(S)=0$,  we have
$\frac{k|S|}{n} - \eps \big(\frac{k|S|}{n} + \frac{k^2}{n}\big) \leq
\Gold(S) \leq \frac{k|S|}{n} + \eps \big(\frac{k|S|}{n} +
\frac{k^2}{n}\big)$. Thus we have
\begin{align*}
(1-\eps')\cov(S) &\leq \frac{1}{1+\eps}\cov(S) \\
&= \frac{1}{1+\eps}\left[k + \frac{n}{k}\Gold(S)\right]
\\&\leq  \frac{1}{1+\eps}\left[k + \frac{n}{k}\left( \frac{k|S|}{n} + \eps \left(\frac{k|S|}{n} + \frac{k^2}{n}\right) \right)\right] 
\\&=  \frac{1}{1+\eps}\left[k +  {|S|} + \eps ({|S|} + {k} )\right] 
\\&= k+|S| = \cov_{\eps'}(S).
\end{align*}
Similarly we have
\begin{align*}
(1+{\eps'})\cov(S)&\geq\frac{1}{1-\eps}\cov(S)\\
&= \frac{1}{1-\eps}\left[k + \frac{n}{k}\Gold(S)\right]
\\&\geq  \frac{1}{1-\eps}\left[k + \frac{n}{k}\left( \frac{k|S|}{n} - \eps \left(\frac{k|S|}{n} + \frac{k^2}{n}\right) \right)\right] 
\\&=  \frac{1}{1-\eps}\left[k +  {|S|} - \eps ({|S|} + {k} )\right] 
\\&= k+|S| = \cov_{\eps'}(S).
\end{align*}
Therefore, $\cov_{{\eps'}}$ is a $(1 \pm \eps')$-approximate oracle to $\cov$.
%\end{proofof}

For an arbitrary subset $S$ of size $k$ with $\Pure_\eps(S)=0$, we
have
\begin{align*}
\frac{\cov(S)}{\opt}
&< \frac{k + \frac{n}{k}\Gold(S)}{n}
\\
&\leq  \frac{k + \frac{n}{k}\Big[ \frac{k|S|}{n} + \eps \Big(\frac{k|S|}{n} + \frac{k^2}{n}\Big) \Big]}{n}
\\
&\leq  \frac{k + \frac{n}{k}\Big[ \frac{k^2}{n} + \eps \Big(\frac{k^2}{n} + \frac{k^2}{n}\Big) \Big]}{n}
\\
&=  \frac{\left(2k + 2\eps k \right)}{n} 
\\
&\leq \frac{4k}{n}.
\end{align*}
Thus, if $S$ is a $\frac{4k}{n}$-approximate solution to the $k$-cover
instance, we have $\Pure_\eps(S)=1$.  Therefore, any
$\frac{6k}{n}$-approximation algorithm returns a set
$S$ such that $\Pure_\eps(S)=1$ with probability at least
$\frac{{6k}/{n}-{4k}/{n}}{\opt}=
\frac{2k/n}{n+k}=\frac{2k}{n^2+kn}\geq \frac{1}{n^2}$.

Recall that for any subset $S$ given that $\Pure_{\eps}(S)=0$, the
value of $\cov_{{\eps'}}(S)$ is predetermined, and can be computed
independent of the actual value of $\cov(S)$. Thus, using a
$\frac{6k}{n}$-approximation algorithm for the $k$-cover problem with
${\eps'}$-error oracle, with probability $\frac{1}{n^2}$, one can find
a set $S$ such that $\Pure_{\eps}(S)=1$, using the same number of queries.
Theorem~\ref{thm:purehard} states that the number of queries is not less than
\begin{align*}
\frac 1 {n^2}\exp\left(\Omega\Big(\frac{{\eps'}^2k^2}{n}\Big)\right) &=
\frac 1 {n^2}\exp\left(\Omega\Big(\frac{\eps^2k^2}{n}\Big)\right) \in
\exp\left(\Omega\Big(\frac{\eps^2k^2}{n}-\log
  n\Big)\right)=\exp\left(\Omega\Big({n\eps^2\alpha^2}-\log n\Big)\right).\qedhere
\end{align*}
\end{proofof}

%% Merge the following into the sketch section in appendix
%\section{Omitted proofs}

%\begin{proofof}{Claim~\ref{clm:delta-noise}}
%\end{proofof}

\section{Omitted proofs for the sketching technique}\label{sec:sketch-app}
\begin{proofof}{Lemma~\ref{lm:str:one}}
  Let $X_u$ be a random variable indicating whether $h(u)\leq p$ for a
  vertex $u\in \Gamma(G,S)$.  By definition we have
  $\cov(S)=|\Gamma(G,S)|$ and $\sum_{u\in
    \Gamma(G,S)}X_u=|\Gamma(H_p,S)|$. Thus, we have
\begin{align*}
\E\Big[|\Gamma(H_p,S)|\Big]=
\E\left[\sum_{u\in \Gamma(G,S)}\hspace{-2mm} X_u\right]=
\hspace{-2mm}\sum_{u\in \Gamma(G,S)} \hspace{-2mm}\E[X_u]=
\hspace{-1mm}\sum_{u\in \Gamma(G,S)}\hspace{-2mm} p=
p |\Gamma(G,S)|.
\end{align*}
By the Chernoff bound to $\Gamma(H_p,S)$ we know that with probability at least
$1-2 \exp\left(-\frac{\epsilon'^2 p |\Gamma(G,S)|}{3} \right)$,
\begin{align*}
\Big||\Gamma(H_p,S)| -p|\Gamma(G,S)|\Big| \leq \eps' p |\Gamma(G,S)|.
\end{align*}
In other words,
\begin{align*}
\Pr\left( \Big|\frac 1 p |\Gamma(H_p,S)| -\cov(S)\Big| \leq \eps' \cov(S) \right) \geq 1-2 \exp\left(-\frac{\eps'^2 p \cov(S)}{3} \right).
\end{align*}
Setting $\eps'=\eps \frac{\opt_k}{\cov(S)}$ in the above yields
\begin{align*}
\Pr\left( \Big|\frac 1 p |\Gamma(H_p,S)| -\cov(S)\Big| \leq \eps \opt_k \right) &\geq 1-2 \exp\left(-\frac{\eps^2 \opt_k^2 p }{3 \cov(S)} \right) \\
&\geq 1-2 \exp\left(-\frac{\eps^2 \opt_k^2 }{3 \cov(S)} \frac{6\delta'}{\eps^2 \opt_k} \right) &\text{from definition of $p$,}\\
&= 1-2 \exp\left(-\frac{ \opt_k }{3 \cov(S)} {6\delta'} \right) \\
&\geq 1-2 \exp\left(-\frac{6\delta'}{3 }  \right) &\text{since $\opt_k\geq \cov(S)$,}\\
&> 1-\exp\left(1-\frac{6\delta'}{3 }\right)  &\text{as $2<e$,}\\
&\geq 1-e^{\delta'}. &\qedhere
\end{align*}
%This completes the proof.
\end{proofof}

\begin{proofof}{Lemma~\ref{lm:str:all}}
  Set $\delta'=k\delta\log n$.  Lemma~\ref{lm:str:one} states that for
  an arbitrary $S\subseteq \sets$ of size at most $k$, we have with
  probability $1-e^{-k \delta\log n}$,
\begin{align*}
\Big|\frac 1 p |\Gamma(H_p,S)| -  \cov(S)\Big| \leq \eps \opt_k.
\end{align*}
Note that there are ${n \choose k}$ different sets $S$ of size $k$.
By the union bound, with probability $1-{n \choose k}e^{-k\delta\log
  n} \geq 1-{n ^ k} e^{-k\delta\log n}= 1- e^{-\delta}$, we have for
all such choices
\begin{align}\label{eq:str:allbound}
\Big|\frac 1 p |\Gamma(H_p,S)| - \cov(S)\Big| \leq \eps \opt_k.
\end{align}
Let $\opt_k$ be the optimum solution on $G$ and let $S$ be the
solution obtained from the $\alpha$-approximation algorithm \Alg when
run on $H_p$. Applying Inequality~\eqref{eq:str:allbound} to $\opt_k$ and $S$, we simultaneously have
\begin{align}\label{eq:str:optgoodH}
\left|\frac 1 p |\Gamma(H_p,\opt_k) | - \opt_k \right| \leq \eps \opt_k
\end{align}
and 
\begin{align}\label{eq:str:Hgood}
\left|\frac 1 p |\Gamma(H_p,S)| - \cov(S)\right| \leq \eps \opt_k.
\end{align}
In addition, since $S$ is an $\alpha$-approximate solution on $H_p$ we
have
\begin{align}\label{eq:str:Sgood}
|\Gamma(H_p,\opt_k)| \leq \frac{1}{\alpha}  |\Gamma(H_p,S)|.  
\end{align}
Inequalities \eqref{eq:str:Sgood} and \eqref{eq:str:optgoodH} together
ensure with probability $1-e^{-\delta}$ that
\begin{align*}
\alpha \opt_k - \frac 1 p |\Gamma(H_p,S)| \leq \alpha \eps \opt_k.
\end{align*}
Combining the above with Inequality \eqref{eq:str:Hgood}, we obtain
\begin{align*}
\alpha \opt_k - \cov(S) \leq \alpha \eps \opt_k + \eps \opt_k \leq 2\eps \opt_k.
\end{align*}
This means that $S$ is an $(\alpha-2\eps)$-approximation to $k$-cover
on $G$ as desired.

%Finally, letting $S = \arg\max_S|\Gamma(H_p,S)|$ in
%Inequality~\eqref{eq:str:Hgood} gives
%\begin{align*}
%\left|\frac 1 p \opt_k(H_p) - \cov(\opt_k(H_p))\right| \leq \eps \opt_k.
%\end{align*}
%This together with $\cov(\opt_k(H_p))\leq \opt_k $,
%$|\Gamma(H_p,\opt_k)|\leq \opt_k(H_p)$ and Inequality
%\eqref{eq:str:optgoodH} shows that
%\begin{align*}
%\left|\opt_k - \frac{1}{p} \opt_k(H_p)\right|\leq 2\eps \opt_k.  &\qedhere
%\end{align*}
\end{proofof}

\begin{proofof}{Lemma~\ref{lm:str:CpN()<}}
By applying Lemma~\ref{lm:str:all} to $S=\arg\max_S |\Gamma(H_p,S)|$,
we have
\begin{align*}
\frac 1 p \max_{S\subseteq \sets:|S|=k}\hspace{-1mm}|\Gamma(H_p,S)| - \cov(S) \leq \eps \opt_k.
\end{align*}
Combining with 
$\cov(S)\leq \opt_k$ and noting that $H'_p$ is a subgraph of $H_p$
gives
%\begin{align}\label{eq:str:Hp(1+eps)opt}
%\max_{S\subseteq \sets:|S|=k}|\Gamma(H_p,S)|\leq  p (1+\eps) \opt_k.
%\end{align}
%Remark that $H'_p$ is a subgraph of $H_p$. Thus for any subset $S\subseteq \sets$ we have $\Gamma(H'_p,S)\leq \Gamma(H_p,S)$. This together with Inequality (\ref{eq:str:Hp(1+eps)opt}) says that 
\begin{align} %\label{eq:str:H'p(1+eps)opt}
 \max_{S\subseteq \sets:|S|=k}|\Gamma(H'_p,S)|&\leq p (1+\eps) \opt_k
\nonumber
\intertext{and we plug in the definition of $p$ to obtain}
%\end{align}
%By definition we have $p= \frac{6 \log(n)k\delta}{\eps^2 \opt_k}$. If we apply this to Inequality (\ref{eq:str:H'p(1+eps)opt}), with probability $1-e^{-\delta}$ we have
%\begin{align*}
 \max_{S\subseteq \sets:|S|=k}\Gamma(H'_p,S) 
&\leq  \frac{6C \log(n)k\delta}{\eps^2 \opt_k} (1+\eps) \opt_k \nonumber\\
 &= \frac{6C (1+\eps)\log(n)k\delta}{\eps^2} \nonumber\\
 &\leq \frac{12C\log(n)k\delta}{\eps^2}.  \nonumber\qedhere
\end{align}
%\end{align*}
%This completes the proof of the lemma.
\end{proofof}

\begin{proofof}{Lemma~\ref{lm:str:N()>}}
  Let $\opt_{H'} = \arg\max_S|\Gamma(H'_p,p)|$. There is a set $v\in
  \opt_{H'}$ such that
  $|\Gamma(H'_p,\opt_{H'})|-|\Gamma(H'_p,\opt_{H'}-{v})|\leq
  \frac{|\Gamma(H'_p,\opt_{H'})|}{k}$, i.e., the marginal effect of
  $v$ is at most a $\frac 1 k$ fraction of the total value of
  $|\Gamma(H'_p,\opt_{H'})|$. Notice that by optimality of
  $\opt_{H'}$, replacing $v$ with any other set does not increase the
  union size.  Thus, for any vertex $v'\in \sets$ the number of
  neighbors of $v'$ in $\elements\setminus \Gamma(H'_p,\opt_{H'})$ is
  at most $\frac{|\Gamma(H'_p,\opt_{H'})|}{k}$. Therefore, the number
  of edges between $\sets$ and $\elements\setminus
  \Gamma(H'_p,\opt_{H'})$ does not exceed
  $n\frac{|\Gamma(H'_p,\opt_{H'})|}{k}$.

  On the other hand, the degree of the elements in
  $\Gamma(H'_p,\opt_{H'})$ is at most $\frac{\log(1/\eps)n}{\eps k}$,
  hence the number of edges between $\sets$ and
  $\Gamma(H'_p,\opt_{H'})$ does not exceed $|\Gamma(H'_p,\opt_{H'})|
  \frac{n\log(1/\eps)}{\eps k}$. Therefore, one can bound the total
  number of edges in $H'_p$ as follows.
\begin{align*}
m'_p &\leq n\frac{|\Gamma(H'_p,\opt_{H'})|}{k} +|\Gamma(H'_p,\opt_{H'})| \frac{n\log(1/\eps)}{\eps k}\\
&\leq \Big|\Gamma(H'_p,\opt_{H'})\Big|\cdot\frac{n}{k} \left(1+ \frac{\log(1/\eps)}{\eps }\right)\\
&\leq \Big|\Gamma(H'_p,\opt_{H'}) \Big|\cdot\frac{n}{k} \cdot\frac{2\log(1/\eps)}{\eps }.
\end{align*}
We obtain by reordering
\begin{align*}
m'_p \frac{\eps k}{2n\log(1/\eps)} \leq |\Gamma(H'_p,\opt_{H'})|.  &\qedhere
\end{align*}
\end{proofof}

\begin{proofof}{Theorem~\ref{thm:str:main}}
  Pick $p\geq \frac{6 k\delta \log n}{\eps^2 \opt_k}$. 
  By Lemma~\ref{lm:str:H'good}, 
  any $\alpha$-approximate solution on $H'_p$ is an
  $(\alpha-\eps)$-approximate solution on $H_p$ with probability
  $1-e^{-\delta}$.
  Moreover, 
  we know from Lemma~\ref{lm:str:all} that 
  any $(\alpha-\eps)$-approximate solution on $H_p$
  is an $(\alpha-3\eps)$-approximate solution on $G$ with probability
  $1-e^{-\delta}$. 
  Therefore, any $\alpha$-approximate solution on
  $H'_p$ is an $(\alpha-3\eps)$-approximate solution on $G$ with
  probability $1-2e^{-\delta}$.

Let us set $p'=\frac{6 k\delta \log n }{\eps^2 \opt_k}$ and 
$p_0=\frac 1 {m},p_1=\frac 1 {m(1-\eps)},p_2=\frac 1 {m(1-\eps)^2}, \dots, p_{\mu}=1$, where $\mu=O(\log m)$.  
Indeed, there is some $i$ such that $p'\leq p_i\leq \frac 1
{1-\eps} p'$. Remark that we set $\delta =
{\log(\log_{1-\eps}m)}\delta''$.  We may assume without loss of
generality that Lemmas~\ref{lm:str:all}, \ref{lm:str:H'good}
and \ref{lm:str:CpN()<} all hold for every $p_j$ with $j\geq i$ since
union bound ensures this outcome happens with probability at least
\begin{align*}
1-3\log_{1-\eps}m \cdot e^{\delta} =
1-3\log_{1-\eps} m\exp[{{\log(\log_{1-\eps}m)}\delta''}] =
1-3e^{\delta''}.
\end{align*}

%%%%%%%%%%%%%%%%%%

Let $p^*$ be (a random number) such that $p^*\geq \frac 1
{1-\eps}\frac{6 k\delta\log n}{\eps^2 \opt_k}$. Remark that, since
$p^*\geq\frac 1 {1-\eps}p'$, there is some (random number) $j$ such
that $p'\leq p_j \leq p^*\leq p_{j+1}= \frac{p_j}{1-\eps}$.  Thus,
\begin{align}\label{X}
\left|\opt_k-\frac{1}{p_j}\opt(H'_{p_j}) \right|\leq 3\eps \opt_k,
\end{align}
and similarly,
\begin{align*}
3\eps \opt_k &\geq 
\left|\opt_k-\frac{1}{p_{j+1}}\opt(H'_{p_{j+1}})\right|\\
&=\left|\opt_k-\frac{1-\eps}{p_j}\opt(H'_{p_{j+1}}) \right|\\
&\geq (1-\eps) \left|\opt_k-\frac{1}{p_j}\opt(H'_{p_{j+1}}) \right|-\eps \opt_k,
\end{align*}
which, assuming $\eps \leq \frac{1}{5}$, gives 
\begin{align}\label{Y}
5\eps \opt_k \geq 
 \left|\opt_k-\frac{1}{p_j}\opt(H'_{p_{j+1}})\right|.
\end{align}
Combining \eqref{X} and \eqref{Y} yields
\begin{align}\label{Z}
\left|\frac{1}{p_j}\opt(H'_{p_{j+1}})-\frac{1}{p_j}\opt(H'_{p_j})\right|\leq 8\eps \opt_k.
\end{align}
The inequalities $p_j \leq p^*\leq p_{j+1}$ implies $H'_{p_{j}} \subseteq H'_{p^*} \subseteq H'_{p_{j+1}}$, hence 
\begin{align}\label{A}
\Gamma(H'_{p_{j}},S) \leq \Gamma(H'_{p^*},S) \leq
\Gamma(H'_{p_{j+1}},S) \qquad \mbox{for any set $S$},
\end{align}
and in turn,
\begin{align}\label{B}
\opt(H'_{p_{j}}) \leq \opt(H'_{p^*}) \leq \opt(H'_{p_{j+1}}).
\end{align}
Combining \eqref{Z} and \eqref{B} gives
\begin{align}\label{C}
\frac{1}{p_j}\opt(H'_{p_{j+1}})-\frac{1}{p_j}\opt(H'_{p^*}) \leq 8\eps \opt_k.
\end{align}
%and also
%\begin{align}\label{D}
%\frac{1}{p_{j+1}}\opt(H'_{p^*})-\frac{1}{p_{j+1}}\opt(H'_{p_j})\leq
%\frac{1}{p_j}\opt(H'_{p^*})-\frac{1}{p_j}\opt(H'_{p_j})\leq 8\eps \opt_k.
%\end{align}
%

Now suppose $S$ is an $\alpha$-approximate solution on $H'_{p^*}$. We have
\begin{align*}
 \cov(S) + \eps \opt_k
 &\geq \frac{1}{p_{j+1}} |\Gamma(H_{p_{j+1}},S)| &\text{from
 	Lemma~\ref{lm:str:all},}
 \\&\geq \frac{1}{p_{j+1}} |\Gamma(H'_{p_{j+1}},S)| &\text{since $H'_{p_{j+1}} \subseteq H_{p_{j+1}}$,}
 \\&\geq \frac{1}{p_{j+1}} |\Gamma(H'_{p^*},S)| &\text{by \eqref{A},}
\\&\geq \alpha \frac{1}{p_{j+1}}  \opt(H'_{p^*}) &\text{from definition of $S$,}
\\&= {\alpha}{(1-\eps)} \frac{1}{p_{j}}  \opt(H'_{p^*}) 
\\&\geq {\alpha}{(1-\eps)}[ \frac{1}{p_{j}} \opt(H'_{p_{j+1}}) -
8\eps \opt_k] &\text{from \eqref{C},}
\\&\geq  \alpha\frac{1}{p_{j+1}} \opt(H'_{p_{j+1}}) - \alpha{8\eps}{(1-\eps)} \opt_k
\\&\geq  \alpha\frac{1}{p_{j+1}} \opt(H'_{p_{j+1}}) - \alpha{8\eps} \opt_k 
\\&\geq  \alpha\frac{1}{p_{j+1}} \opt(H'_{p_{j+1}}) - {8\eps} \opt_k
&\text{since $\alpha \leq 1$,}
\\&\geq {\alpha} \opt_k - 11\eps \opt_k 
\\&= (\alpha-11\eps) \opt_k,
\end{align*}
that is, any $\alpha$-approximate solution on $H'_{p^*}$ is an
$(\alpha - 12\eps)$-approximate solution on $G$.

Finally we argue that $p^*\geq \frac 1 {1-\eps}\frac{6
  k\delta\log n}{\eps^2 \opt_k}$ for $H'_{p^*}= H_{\leq n}$.  We
set $C=\frac 1 {1-\eps}$ in Lemma~\ref{lm:str:CpN()<} to obtain
$\max_{S\subseteq \sets:|S|=k}|\Gamma(H'_{p''},S)|\leq \frac{\frac{12k\delta}
  {1-\eps}\log n}{\eps^2}$, where $p''=\frac 1 {1-\eps}\frac{6
  k\delta\log n}{\eps^2 \opt_k}$. 
Thus, $H'_{p^*}$ contains $H'_{p''}$ if $\max_{S\subseteq
  \sets:|S|=k}|\Gamma(H'_{p^*},S)|\geq \frac{\frac {12k\delta}
  {1-\eps}\log n}{\eps^2}$.
On the other hand if we set $m'_{p^*}\geq \frac{24 n\delta\log(1/\eps)
  \log n}{(1-\eps)\eps^3}$ in Lemma~\ref{lm:str:N()>}, we get
\begin{align*}
|\Gamma(H'_{p^*},\opt_{H'})| &\geq
m'_{p^*} \frac{\eps k}{2n\log(1/\eps)} 
\\&\geq \frac{24 n\delta\log(1/\eps) \log n}{(1-\eps)\eps^3}\cdot 
\frac{\eps k}{2n\log(1/\eps)}
\\&= \frac{\frac {12 k\delta} {1-\eps}\log n}{\eps^2}.  \qedhere
\end{align*}
\end{proofof}

\begin{proofof}{Lemma~\ref{lm:loglambda}}
  Let $\opt_{k'}$ be the set cover of size $k'$ on the input graph.
  Pick $p^*$ such that $H'_{p^*}=H_{\leq n}$. Remark that $H_{p^*}$ is
  an induced subgraph of $G$ containing all sets \sets. Thus,
  $\opt_{k'}$ is a set cover in $H_{p^*}$, as well.

  Similarly to the proof of Lemma~\ref{lm:str:H'good}, we present here
  a randomized solution $R^*$ that, in expectation, covers $1-\xi
  \eps$ fraction of the vertices. This implies that there exists a
  solution of size $k'$ covering at least $1-\xi \eps$ fraction of
  the elements.

  We construct $R^*$ by removing $\xi \eps k'=\eps k$ sets chosen
  uniformly at random from $\opt_{k'}$ and adding $\xi \eps k'=\eps k$
  other sets picked uniformly at random.

  Each element with degree at most $\frac{n\log(1/\eps)}{\eps k}$ in
  $H_{p^*}$ still exists in $\Gamma(H_{p^*},R^*)$ with probability
  $1-\xi \eps$, hence it is in $\Gamma(H_{\leq n},R^*)$, too. On the
  other hand, for any element $u$ with degree at least
  $\frac{n\log(1/\eps)}{\eps k}$ in $H_{p^*}$, the probability that
  $u$ is not contained by any of $\xi \eps k'$ randomly chosen sets
  cannot exceed
\begin{align*}
\left(1-\frac{\frac{n\log(1/\eps)}{\eps k}}{n}\right)^{\eps k}
=\left(1-\frac{\log(1/\eps)}{\eps k}\right)^{\eps k}
\leq \left(\frac 1 e\right)^{\log{\frac{1}{\eps}}}
= \eps.
\end{align*}
Thus, each element of $H_{p^*}$ exists in $\Gamma(H_{\leq n},R^*)$, as
well, with probability $1-\eps$.
\end{proofof}

\section{Omitted proofs for the streaming setting}\label{sec:alg-app}

\begin{proofof}{Theorem~\ref{thm:alg:kcover}}
  The first part of the statement is derived from the approximation
  guarantee of \greedy and Theorem~\ref{thm:str:main}.  These two
  imply that $\greedy(k,H_{\leq n}(k,\eps',\delta''))$ is a $1-\frac 1
  e -12 \eps' = 1-\frac 1 e -\eps$ approximate solution with
  probability $1-3e^{\delta''}= 1-3e^{2+\log n}\geq 1-\frac 1 n$, as
  desired.

  By definition of $H_{\leq n}(k,\eps',\delta'')$, the number of edges
  in this sketch is not more than
\begin{align*}
&\frac{24 n\log(1/\eps') \log(n){\delta''}{\log\log_{1-\eps'} m}}{(1-\eps')\eps'^3}+n
\\=&
\frac{24 n\log(\frac {12}{\eps} ) \log n{(2+\log n)}{\log\log_{1- {\eps}/ {12}}m)}}{(1-{\eps}/{12})({\eps}/ {12})^3}+n
\\\in\:&
O\Big(\frac {n \log^3n\log^2m\log\log m}{\eps'^3}\Big) \in \tilde{O}(n) 
&\text{w.l.o.g.~assuming $\eps' \in \Omega(\frac 1 m).$} &\qedhere
\end{align*}
\end{proofof}

\begin{lemma}\label{lm:H<n:1-l-eps}
	Let $k'$ be the size of a set cover of graph $G$. Pick arbitrary
	$\eps\in (0,1]$, $\lambda'\in (0,\frac 1 e]$, and let
	$k=\log(\frac{1}{\lambda'})k'$. Then $\greedy(k,H_{\leq n})$ covers at
	least $1-\lambda'- \eps\log\frac{1}{\lambda'}$ fraction of elements
	\elements in $H_{\leq n}$.
\end{lemma}

\begin{proof} %{Lemma~\ref{lm:H<n:1-l-eps}}
  Lemma~\ref{lm:loglambda} ensures the existence of a solution of
  size $k'$ on $H_{\leq n}$ covering at least
  $1-\eps\log\frac{1}{\lambda'}$ fraction of \elements in
  $H_{\leq n}$.  Thus, $\greedy(k,H_{\leq n})$ covers at least
\begin{align*}
  (1-\lambda')\left[1-\eps\log\frac{1}{\lambda'}\right]\geq
  1-\lambda'-\eps\log\frac{1}{\lambda'}
\end{align*}
fraction of \elements in $H_{\leq n}$. 
\end{proof}

\begin{proofof}{Lemma~\ref{lm:epsSC0}}
  According to Lemma~\ref{lm:H<n:1-l-eps}, for a graph $G$ with a set
  cover of size $k'$, the solution $\greedy(k,H_{\leq n})$ covers at
  least $1-\lambda'-\eps\log\frac{1}{\lambda'}$ fraction of \elements
  in $H_{\leq n}$, hence Algorithm~\ref{Alg:main:epsSetCover0} does
  not return false, as claimed.

  We prove the second part of the lemma in the case Theorem
  \ref{thm:str:main} holds for $H_{\leq n}$.  This happens with
  probability $1-3e^{\delta''}= 1-3e^{\log(C'n)+2}\geq 1-\frac 1
    {C'n}$.

    On the other hand, in case Algorithm~\ref{Alg:main:epsSetCover0}
    does not return false, $\greedy(k,H_{\leq n})$ covers no less than
    $1-\lambda'-\eps\log\frac{1}{\lambda'}$ fraction of \elements in
    $H_{\leq n}$. Then Theorem~\ref{thm:str:main} implies that
    $\greedy(k,H_{\leq n})$ covers at least
\begin{align*}
  1-\lambda'- \eps\log\frac{1}{\lambda'} - 12 \eps \geq 1-\lambda'-13
  \eps\log\frac{1}{\lambda'} = 1-\lambda'-\eps'
\end{align*}
fraction of elements.

By definition of $H_{\leq n}(k,\eps,\delta'')$, the number of edges in
this sketch is at most
\begin{align*}
  &\frac{24 n\log(1/\eps)
    \log(n){\delta''}{\log\log_{1-\epsilon}m}}{(1-\eps)\eps^3}+n
  \\
  =\:&
  \frac{24 n\log(1/\eps) \log
    n\log_{1+\eps}n[\log(C'n)+2]}{\log\log_{1-\eps}m}{(1-\eps)\eps^3}+n
      \\
      \in\:&
      O\Big(\frac {n \log^2n\log^2m\log C'\log\log m}{\eps^3}\Big)
      &\text{w.l.o.g.~assuming $\eps \in \Omega(\frac 1 m)$,}
      \\
= \:&
      O\left(\frac {n \log^2n\log^2m\log C'\log\log m}{\left[\frac
        {\eps'} {13 \log(1/ {\lambda'})}\right]^3}\right) 
\\
=\: & O\left(\frac {n\log^2n\log^5m\log C'\log\log m}{ {\eps'} ^3}\right)
      &\text{w.l.o.g.~assuming $\lambda' \in \Omega(\frac 1 m)$,}
      \\
\subseteq\:& O\left(\frac {n \log^2n\log^6m\log C'}{ {\eps'}^3}\right).  &\qedhere
\end{align*}
\end{proofof}

\begin{proofof}{Theorem~\ref{thm:alg:epsSetCover}}
  Remark that each iteration of the loop in
  Algorithm~\ref{Alg:main:epsSetCover} increases $k'$ by a factor of
  $1+\frac{\eps}{3}$, and we always keep $k'\leq n$. Thus, we run at
  most $\log_{1+\frac{\eps}{3}}n$ instances of
  Algorithm~\ref{Alg:main:epsSetCover0}. Lemma~\ref{lm:epsSC0} holds
  for each of these instances with probability $\frac{1}{C
    n\log_{1+\frac{\eps}{3}}n}$. They all hold with probability
  $1-\frac{1}{Cn}$ by the union bound. We prove the statement of the
  theorem assuming this.

  Let $k^*$ be the size of the minimum set cover in $G$, and let $k'$
  be the final value of this variable after run of
  Algorithm~\ref{Alg:main:epsSetCover}. Indeed
  Algorithm~\ref{Alg:main:epsSetCover0} returns false for
  $\frac{k'}{1+\eps/3}$, hence
%  $\frac{k'}{1+\varepsilon/3}\leq k^*$, or equivalently 
  $k'\leq (1+\eps/3)k^*$. Note that, the size of the set returned by
  Algorithm~\ref{Alg:main:epsSetCover} is
\begin{align*}
k'\log\frac{1}{\lambda' }&=k'\log\frac{1}{\lambda e^{-\varepsilon/2}}
\\&= k'\left[\log\frac{1}{\lambda}+\frac{\eps}{2}\right]
\\&\leq \left[\log\frac{1}{\lambda}+\frac{\eps}{2}\right] \cdot
\left[1+\frac{\eps}{3}\right]\cdot k^*
\\&= k^* \left[\log\frac{1}{\lambda}+\frac{\eps}{2}+ \frac{\eps}{3} \log\frac{1}{\lambda}+\frac{\eps^2}{6}\right]
\\&\leq (1+{\eps}) k^*\log\frac{1}{\lambda}.
\end{align*}
On the other hand, this solution covers at least
\begin{align*}
1-\lambda'-\eps' = 1- \lambda e^{-\eps/2} - {\lambda}(1-e^{-\eps/2}) = 1-\lambda
\end{align*}
fraction of the vertices in $G$, as claimed.

Lemma~\ref{lm:epsSC0} bounds the number of edges in the sketch used by
each instance of Algorithm~\ref{Alg:main:epsSetCover0} as
\begin{align*}
O\left(\frac {n \log^2n\log^6m\log C'}{ {\eps'} ^3}\right)
=
O\left(\frac {n \log^2n\log^6m\log C \log_{1+\frac{\eps}{3}}n}{ [{{\lambda}(1-e^{-\eps/2})}] ^3}\right)
\subseteq
\tilde{O}(n/\lambda^3) \subseteq\tilde{O}_{\lambda}(n). 
\end{align*}
With $\log_{1+\eps/3}n$ runs of Algorithm~\ref{Alg:main:epsSetCover0},
the total number of edges in all the sketches used in this algorithm
is $\tilde{O}(n/\lambda^3) \subseteq\tilde{O}_{\lambda}(n)$.
\end{proofof}

%\subsection{The distributed model}
%\input{MapReduce.tex}

%\subsection{RAM model}
%\input{RAM.tex}

\section{An $O(nk)$ sketch using $\ell_0$ sketches}\label{Apx:O(nk)}

We use the $\ell_0$ sketch defined as follows.
\begin{definition}[$\ell_0$ sketch]
Given a multiset $\Pi$, the number of distinct elements in $\Pi$ is said to be $\ell_0$ of $\Pi$.
\end{definition}
Cormode et al.~\cite{cormode2003comparing} provides an $O(\log n\frac{1}{\eps^2}\log\frac 1 {\delta})$ space streaming algorithm that with probability $1-\delta$ gives a $1-\eps$ approximation to the $\ell_0$ sketch. Interestingly, one can merge two of these $\ell_0$ sketches and again with probability $1-\delta$ get a $1-\eps$ approximation to the $\ell_0$ sketch of the merged multiset.
 
Given an input graph $G(\sets,\elements)$, we maintain an $\ell_0$ sketch (using the algorithm in \cite{cormode2003comparing}) for the set of neighbors of each vertex in $\sets$ (we fix $\eps$ and $\delta$ used in the $\ell_0$ sketch later). We estimate the value of each set $S\subseteq \sets$ by merging the $\ell_0$ sketches corresponding to the vertices in $S$.

Consider that, for each set $S\subseteq \sets$, with probability $1-\delta$, we give a $1-\eps$ approximation of the coverage valuation of $S$. Nevertheless, to find the $k$-cover solution, we look at ${n \choose k}$ sets. The union bound ensures that, with probability $1-{n \choose k} \delta$, all the $n \choose k$ estimations are accurate. We set $\delta=\frac{1}{\tilde{\Theta}({n \choose k})}$, so these hold with high probability. 

The space required by each $\ell_0$ sketch is $O(\log n\frac{1}{\eps^2}\log\frac 1 {\delta}) = \tilde{O}(\log{n \choose k})= \tilde{O}(k)$. The space required by the algorithm to maintain $n$ sketches is thus $\tilde{O}(nk)$, giving the following theorem.

\begin{theorem} 
  Using $\ell_0$ sketches, there exists an exponential-time $1-\eps$ approximation streaming algorithm for $k$ cover using $\tilde{O}(nk)$ space.
  \end{theorem}

%%%%%%%%%%%%%%%%%%%%%%%%%%%%%%%%%%%%%%%%%%%%%%%%%%%%%%%%%%%%
%%%%%%%%%%%%%%%%%%%%%%%%%%%%%%%%%%%%%%%%%%%%%%%%%%%%%%%%%%%%
%%%%%%%%%%%%%%%%%%%%%%%%%%%%%%%%%%%%%%%%%%%%%%%%%%%%%%%%%%%%

\section{Hardness of streaming problems} \label{sec:strHard}

\xxx{I should probably go through this once more.}
% % % % % % % % % % % % % % % % % % % % % % % % %

Here, we give a lower bound on the space required to solve $k$-cover in the streaming setting. To establish this lower bound, we present a reduction from the set-disjointness problem. In the set disjointness problem, two parties, namely Alice and Bob, each holds a subset of ${1,2,\dots,n}$. The goal is to determine whether the sets are disjoint or not. Razborov~\cite{razborov1992distributional} and Kalyanasundaram and Schintger~\cite{kalyanasundaram1992probabilistic} provide a lower bound of $\Omega(n)$ even when allowing randomization.

\begin{proofof}{Theorem \ref{thm:str:hard}}
Let $A$ be the set that Alice holds, and let $B$ be the set that Bob holds. In our hard $k$ cover instance we have two vertices (namely $a$ and $b$) in $\elements$ and $n$ vertices in $\sets$. Vertex $a$ has an edge to the $i$-th vertex in $\elements$ if and only if $i$ exists in $A$. Similarly, $b$ has an edge to the $i$-th vertex in $\elements$ if and only if $i$ exists in $B$. In the input stream first we see the edges of $a$ (which contains the information Alice holds) and then the edges of $b$ (which contains the information Bob holds).

In this example, if the sets $A$ and $B$ are disjoint, each of the vertices in $\sets$ covers at most one of $a$ and $b$, and thus, the value of an optimum solution to $1$-cover on this graph is $1$. Otherwise, there is a vertex $i$ which has edge to both $a$ and $b$, and thus, the value of an optimum solution to $1$-cover on this graph is $2$. 
Therefore, distinguishing between the case that the value of the optimum solution to $1$-cover is $1$ and the case that this value is $2$ requires $\Omega(n)$ space in total.
\end{proofof}

\end{document}